\documentclass[final]{statsoc}
\pdfoutput=1    

\usepackage[a4paper]{geometry}
\usepackage{amsmath}
\usepackage{amsfonts}
\usepackage{amssymb}

\usepackage[T1]{fontenc}

\usepackage{booktabs}
\usepackage{natbib}
\usepackage{macros}
\usepackage{hyperref}
\usepackage{pdfpages}

\DeclareGraphicsExtensions{.png,.pdf} 

\newcommand{\qedhere}{}

\title[Fast Moment-Based Estimation for Hierarchical Models]{%
  Fast moment-based estimation for hierarchical models%
}

\author[P.\ O.\ Perry]{Patrick O.\ Perry}

\address{Stern School of Business, New York University, USA}

\received{April 2014}
\revised{November 2015}

\runauthor{%
    P.\ O.\ Perry
}

\coaddress{
Patrick O.\ Perry,
Information, Operations, and Management Sciences Department,
Stern School of Business,
New York University,
44 West 4th St,
New York, NY 10012,~USA
}
\email{pperry@stern.nyu.edu}

\begin{document}

\begin{abstract}
Hierarchical models allow for heterogeneous behaviours in a population while
simultaneously borrowing estimation strength across all subpopulations.
Unfortunately, existing likelihood-based methods for fitting hierarchical
models have high computational demands, and these demands have limited their
adoption in large-scale prediction and inference problems.  This paper
proposes a moment-based procedure for estimating the parameters of a
hierarchical model which has its roots in a method originally introduced by
Cochran in 1937.  The method trades statistical efficiency for computational
efficiency.  It gives consistent parameter estimates, competitive prediction
error performance, and substantial computational improvements.  When applied
to a large-scale recommender system application and compared to a standard
maximum likelihood procedure, the method delivers competitive prediction
performance while reducing the sequential computation time from hours to
minutes.

\end{abstract}

\section{Introduction}
\label{sec:intro}

Hierarchical models are appropriate when we collect data
from multiple sub-populations or groups, each of which exhibits different
associations between the measured variables.  Each group can be a particular
classroom, firm, city, time period, or any member of a class of similar
entities.  Rather than ignoring the subpopulation structure and assuming that
all observations are independent, a hierarchical model accounts for the
dependence of the observations within a group by allowing for random
subpopulation-specific effects.  These models and more general mixed models
are widely applied in the natural and social sciences, and many reference
books describe them in detail \citep{Snij12,Scot13}.

By explicitly allowing for between-group variability, hierarchical models hold
two main advantages over models that do not.  First, in accounting for this
variability, a hierarchical model is able to give more accurate uncertainty
estimates for population parameter estimates \citep{Rao65}.  Second, by
drawing strength across similar experimental units, a hierarchical model can
give better group-specific predictions \citep{Rein85}.  The latter phenomenon
is closely related to the performance of Stein's shrinkage estimators
\citep{Morr83}.

One seemingly-appropriate application for hierarchical models is in
recommender systems, where the goal is to take historical data about users,
items, and user ratings of these items to learn users' preferences and to make
recommendations based on these preferences \citep{Adom05}.  Here, users
correspond to groups, and user-specific preferences correspond to random
effects.  In fact, early in the development of recommender systems,
\citet{Cond99} and \citet{Ansa00} advocated for the use of these models and
more general mixed models due to their potential to combine content-based
filtering (recommending based on item-specific attributes) and collaborative
filtering (recommending based on preferences of similar users).

Despite their advantages, in the late 2000s, many authors deemed the
computational costs required to fit a hierarchical model to be prohibitively
high for recommender systems and other similar applications in
commercial-scale settings \citep{Zhan07,Agar08,Naik08,Agar09}.  Most methods
for fitting these models and related factor models are iterative, with a high
computational cost for each iteration.  Letting $q$ denote the number of fixed
and random effects in the model, methods based on expectation-maximization
\citep{Demp81,Zhan08,Agar09}, variational approximations \citep{Arma11},
likelihood maximization \citep{Gold86,Jenn86,Long87,Lind88}, and profile
likelihood maximization, require initial computation costs proportional to $N
q^2$, where $N$ is the number of samples, followed by a series of iterations,
each with computational costs proportional $M q^3$ or $M q^4$, where $M$ is
the number of groups.  This can be substantial when~$M$ and~$N$ are both
large.

In cases where the predictors are sparse, it is possible to exploit this
structure to achieve speed-ups on the order of $q$ or $q^2$, which can be
dramatic if $q$ is large \citep{Zhan07}.  This, however, requires special
structure in the predictor matrices and imposes sparsity constraints on the
parameter estimates.

In general situations, one can partition the data between multiple processors,
compute separate parameter estimates for each chunk, and then combine the
results \citep{Huan05,Gebr12,Khan13,Scot13B}.  These splitting strategies
often require the same total computational cost, but they split the costs
between $K$ processors, reducing wall clock time by a factor of $K$.
An alternative approach is to approximate the data likelihood using a form of
$h$-likelihood and then optimize the resulting criterion via stochastic
gradient descent \citep{Kore09b,Dror11}.  This requires a series of iterations, each
with computation costs proportional to~$N q$, often leading to a lower overall
fitting time.

In this report, we propose an alternative approach, revisiting and extending
a moment-based estimation procedure originally due to \citet{Coch37}.  In this
approach, we fit group-specific estimates in isolation, then combine these
estimates to get population parameter estimates by matching moments.  The main
advantage of the approach over existing alternatives is that it is not
iterative.  There is an initial cost proportional to $N q^2$, followed by a
fixed cost proportional to $M q^4$.  Due to memory locality, in practice the
dominant cost is often proportional to~$M$.  The procedure can be trivially
distributed across $K$ processors, reducing computation by a factor of~$K$.

\begin{figure}
  \centering
  \includegraphics[width=13.5cm]{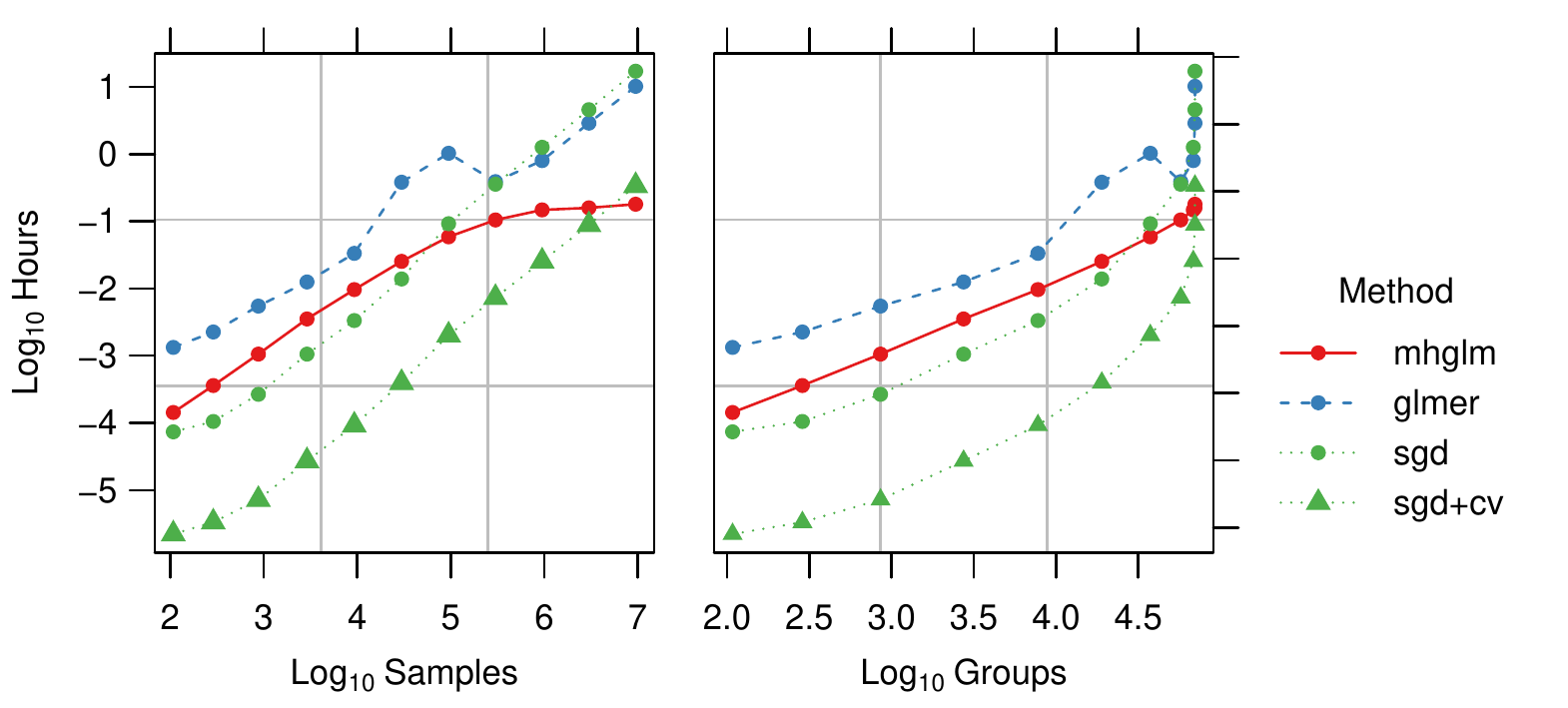}
  \caption{Computational scaling properties for hierarchical model fitting procedures.}
  \label{fig:time-order}
\end{figure}

Fig.~\ref{fig:time-order} demonstrates the potential advantages of the
moment-based estimation method.  This figure shows the amount of CPU time
required by three different procedures---maximum likelihood (\emph{glmer}),
stochastic gradient descent (\emph{sgd}), and the proposed method
(\emph{mhglm})---fitting hierarchical models to subsets of the MovieLens~10M
recommender system dataset \citep{Movi09}.  The first two methods are
implemented in a mix of R, C, and C++; the proposed method is implemented in
R.  In this example, the computational costs required for the first two
methods appear to scale linearly with the sample size, $N$, while for the
latter, the dominant computational costs appear to be proportional to $M$.  At
the largest value of $N$ reported, the proposed method is 50 times faster than
\emph{glmer}, and $1.7$ times faster than \emph{sgd} ($90$ times faster if we
include the cross-validation time required to choose the tuning parameter for
\emph{sgd}).  Notably, even if \emph{glmer} were split across 10 processors,
running the proposed method on a single CPU would still be faster by a factor
of 4.

In this report, we demonstrate that the proposed moment-based estimation
procedure is often faster than likelihood-based methods.  The improvements in
computational efficiency do not come free; they are paid for by sacrificing
some statistical efficiency.  In many large-sample regimes, the loss in
statistical efficiency is small or modest, and it becomes worthwhile to make
this statistical-computational trade-off.

We introduce hierarchical models in more detail in
Section~\ref{sec:hierarchical-models}.  Next, in Section~\ref{sec:moment-est}
we describe the proposed moment-based fitting procedure.  This procedure
depends on a choice of weights, which we discuss in
Section~\ref{sec:weight-choices}. In Sections~\ref{sec:finite-sample-prop}
and~\ref{sec:asymptotics} we derive finite-sample and asymptotic properties
for the estimators, including consistency, relative efficiency, and asymptotic
normality.  We investigate performance in simulations in
Section~\ref{sec:simulations}.  Finally, we apply the method to a recommender
system application in Section~\ref{sec:application}, and close with a brief
discussion in Section~\ref{sec:discussion}.  The on-line supplementary
material contains Appendices~\ref{sec:appendix-first}--\ref{sec:appendix-last}
with additional details and technical lemmas.

The proposed method is implemented in the \texttt{mbest} R
package, available at \url{http://cran.r-project.org/web/packages/mbest/}.
Data and software to generate the figures in this paper are available at
\url{http://ptrckprry.com/reports/}.

\section{Hierarchical models}
\label{sec:hierarchical-models}


Consider a collection of $M$ subpopulations or groups.  In group~$i$ we
observe~$n_i$ random response values denoted individually as $y_{ij} \ (j = 1,
\dotsc, n_i)$, or jointly as the vector $\vy_i$ with $j$th component equal to
$y_{ij}$ for $j = 1, \dotsc, n_i$.  The total number of observations is $N =
\sum_{i=1}^{M} n_i$.  Suppose that each observation $y_{ij}$ has two
associated predictor vectors: a vector~$\vx_{ij}$ of dimension~$p$, and a
vector $\vz_{ij}$ of dimension~$q$.  In matrix form, let $\mX_i$ and $\mZ_i$
be the corresponding predictor matrices of dimensions $n_i \times p$ and $n_i
\times q$, with row~$j$ equal to $\vx_{ij}$ or $\vz_{ij}$, respectively, for
$j = 1, \dotsc, n_i$.  Our goal will be to use the $N$ observations to
estimate the association between the response $y_{ij}$ and the feature vectors
$\vx_{ij}$ and $\vz_{ij}$.

In a hierarchical linear model, we posit that conditional on a vector $\vu_i$
of group-specific random effects, the expectation of the response vector is
determined by the relation
\begin{equation}\label{eqn:hlm-link}
    \E(\vy_{i} \mid \vu_i) = \mX_i \vbeta + \mZ_i \vu_i,
\end{equation}
where $\vbeta$ is a vector of $p$ fixed population effects shared across all
$M$ groups.  Further, we assume that within each group the response values are
independent, with conditional variances given by
\(
    \var(y_{ij} \mid \vu_i) = \sigma^2.
\)
Lastly, we take the random effect vectors $\vu_1, \dotsc, \vu_M$ to be independent and
identically distributed with mean zero and covariance matrix
\(
    \cov(\vu_i) = \mSigma
\)
for some positive-semidefinite matrix $\mSigma$.

Hierarchical generalized linear models are natural extensions of hierarchical
linear models that allow for non-linear relations between the response and the
effects~\citep{Lee96}.  The set-up is similar to that for a hierarchical
linear model, but we replace the relation~\eqref{eqn:hlm-link} with the
nonlinear relation
\(
  \E(\vy_i \mid \vu_i) = g_i^{-1}(\mX_i \vbeta + \mZ_i \vu_i)
\)
for some specified link function $g_i$.  Instead of a variance parameter $\sigma^2$,
we have a dispersion parameter $\phi$ (possibly known).



For a hierarchical linear model or hierarchical generalized linear model,
given observations $\vy = (\vy_1, \dotsc, \vy_M)$ our main inferential task is
estimating the population parameters $\vbeta$, $\mSigma$, and $\sigma^2$.
Once these estimates have been obtained, they can be used together with the
data to estimate (formally, predict) the random effect vectors $\vu_1, \dotsc,
\vu_M$, typically using a Gaussian approximation to the conditional
distribution $\vu_i \mid \vy_i$ with plug-in estimates for quantities
involving $\vbeta$, $\mSigma$, and $\phi$.  In turn, the estimated effect
vectors can be used to forecast future response values.

Our primary focus in this report is developing a computationally efficient
method for estimating $\vbeta$, $\mSigma$, and $\sigma^2$.  We focus on
applications where the number of groups, $M$, is large, with a small or
moderate number of predictors ($p + q \ll M$).

\section{Moment-based estimation}
\label{sec:moment-est}

\subsection{Overview}
\label{sec:moment-based-est}

%

Before likelihood-based fitting procedures for hierarchical models became
ubiquitous, Cochran developed a moment-based approach for fitting a univariate
($p = q = 1$) hierarchical linear model \citep{Coch37,Yate38,Coch54}.  The
method takes group-specific estimates of the effects and then uses weighted
moments of these estimates to approximate the population parameters.
\citet{Swam70} extended Cochran's method to multivariate settings, and
\citet{Cox02} further extended it to allow for hierarchical nonlinear models.
The main advantage of these moment-based estimation methods is that they are not
iterative.  For these methods, and for the extension we introduce, there is a
computational cost of roughly $O\{N(p+q)^2\}$ to fit the initial
group-specific estimates, followed by a cost of $O\{M(p+q)^3 + M q^4\}$ to
combine them.  Furthermore, most of the operations are embarrassingly
parallel, in the sense that it is trivial to split them across multiple
processors.

Moment-based estimation methods for hierarchical models are simple and
computationally efficient.  Unfortunately, existing moment-based approaches
require that $\mX_i = \mZ_i$ for $i = 1, \dotsc, M$.  Moreover, they
require each predictor matrix $\mX_i$ to have full rank.  These restrictions
seem innocuous, but they become prohibitive in many large scale estimation
problems, including the recommender system application discussed in
Section~\ref{sec:application}.  This motivates us to introduce an alternative
extension of Cochran's method, similar in spirit to Swamy's procedure, but
allowing for arbitrary fixed effects and removing most restrictions on the
ranks of the predictor matrices.

\subsection{Intuition from the hierarchical linear model}
\label{sec:intuition}

To gain an intuition into our procedure, we start by considering the
hierarchical linear model.  For $i = 1, \dotsc, M$ define feature matrix
$\mF_i = [ \mX_i\ \mZ_i ]$ of size $n_i \times (p + q)$ and effect vector
$\veta_i = [ \vbeta^\trans\ \vu_i^\trans ]^\trans$ of dimension $(p + q)$.
The first $p$ components of~$\veta_i$ are shared across all $M$
groups, and the last $q$ components are random and specific to group~$i$.
The group-specific response vector can be expressed as
\[
  \vy_i = \mF_i \veta_i + \vvarepsilon_i,
\]
where $\vvarepsilon_i$ has mean zero and is independent of $\vu_i$.

Define the least squares estimate
\[
  \vheta_i = (\mF_i^\trans \mF_i)^{\pinv} \mF_i^\trans \vy_i,
\]
where $^\pinv$ denotes Moore-Penrose pseudo-inverse.  Previous approaches
required $\mF_i$ to have full column rank, but we make no such restriction.
Notably, without this restriction it will not generally be the
case that $\E(\vheta_i \mid \vu_i) = \veta_i$.  Rank degeneracy leads to
aliasing in the coefficients, which precludes unbiased estimation.

Despite potential aliasing, the estimate $\vheta_i$ still contains information
about the effects in the subspace spanned by the rows of $\mF_i$.
Specifically, let
\[
  \mF_i = \mU_i \mD_i \mV_i^\trans
\]
be a compact singular value decomposition, where $\mD_i \succ 0$ is diagonal
with dimension $r_i \times r_i$ and
$\mU_i^\trans \mU_i = \mV_i^\trans \mV_i = \mI_{r_i}$.  
Let $\mV_{i1}$ and $\mV_{i2}$ (dimensions $p \times r_i$ and $q \times r_i$)
contain the first $p$ and last $q$ rows of $\mV_i$, respectively, so that
\[
  \mX_i = \mU_i \mD_i \mV_{i1}^\trans,
  \qquad
  \mZ_i = \mU_i \mD_i \mV_{i2}^\trans,
\]
with $\mV_{i1}^\trans \mV_{i1} + \mV_{i2}^\trans \mV_{i2} = \mI_{r_i}$.
Then,
\begin{subequations}\label{eqn:conditional-moment-relations}
\begin{align}
  \E(\mV_i^\trans \vheta_i \mid \vu_i)
    &= \mV_{i1}^\trans \vbeta + \mV_{i2}^\trans \vu_i, \\
  \cov(\mV_i^\trans \vheta_i \mid \vu_i)
    &= \phi \mD_i^{-2},
\end{align}
\end{subequations}
where $\phi = \sigma^2 = \var(\varepsilon_{ij})$.
Hence, the unconditional expectation and covariance of the effect
components orthogonal to the nullspace of $\mF_i$ are
\begin{align*}
  \E(\mV_{i}^\trans \vheta_i)
    &= \mV_{i1}^\trans \vbeta, \\
  \cov(\mV_i^\trans \vheta_i)
    &= \mV_{i2}^\trans \mSigma \mV_{i2} + \phi \mD_i^{-2}.
\end{align*}
In Section~\ref{sec:estimation-procedure}
we show how to use these moment relations to estimate the model parameters.

For the dispersion parameter, we will use the unbiased estimator
\[
  \hat \phi
  =
  \hat \sigma^2
    =
    \frac{1}{N - \rho} \sum_{i=1}^{M} \norm{\vy_i - \mF_i \vheta_i}^2,
\]
where $\norm{\cdot}$ denotes Euclidean norm and $\rho = \sum_{i=1}^{M} r_i$.
As long as $n_i > r_i$ for at least one group~$i$, this estimator is
well-defined.

\subsection{The general procedure}
\label{sec:estimation-procedure}

We define the general estimation procedure without reference to the response,
the predictor matrices, or the specific data-generating mechanism.  As a
starting point, we will suppose that we have the following:
\begin{enumerate}
\item random effects $\vu_1, \dotsc, \vu_M$ that are independent with mean zero and
covariance matrix $\mSigma$;

\item group specific effect estimates $\vheta_1, \dotsc, \vheta_M$ 
that satisfy the conditional moment relations~\eqref{eqn:conditional-moment-relations};

\item  matrices $\mD_i$ and
$\mV_i = [\mV_{i1}^\trans\ \mV_{i2}^\trans]^\trans\ (i = 1, \dotsc, M)$,
where $\mV_i$ has $r_i$ orthonormal columns, and $\mD_i$ is a symmetric
positive-definite matrix (not necessarily diagonal);

\item dispersion estimate $\hat \phi$ that has expectation $\phi$.
\end{enumerate}
The procedure depends on a choice of symmetric positive-definite weight
matrices, denoted $\mW_1, \dotsc, \mW_M$, where $\mW_i$ has dimension $r_i
\times r_i$.  We will discuss choices for the weights in
Section~\ref{sec:weight-choices}, but for now, take them to be arbitrary.

We will use the weights to combine the group-specific estimates into an estimate for
the fixed effect $\vbeta$.  To do so, define
\begin{equation}\label{eqn:Omega1}
  \mOmega = \sum_{i=1}^{M} \mV_{i1} \mW_{i} \mV_{i1}^\trans.
\end{equation}
If $\mOmega$ is invertible, then we can define a moment-based estimator
for $\vbeta$:
\begin{equation}\label{eqn:beta-hat}
  \vhbeta_{\mW}
    =
    \mOmega^{-1}
    \sum_{i=1}^{M} \mV_{i1} \mW_i \mV_i^\trans \vheta_i.
\end{equation}
By construction, $\vhbeta_{\mW}$ is an unbiased estimator for $\vbeta$.

To introduce an estimator for the random effect covariance matrix $\mSigma$,
first define the matrix-valued function
\[
  \mhA(\vb)
    =
    \sum_{i=1}^{M}
      \mV_{i2} \mW_i
      (\mV_i^\trans \vheta_i - \mV_{i1}^\trans \vb)
      (\mV_i^\trans \vheta_i - \mV_{i1}^\trans \vb)^\trans
      \mW_i \mV_{i2}^\trans.
\]
Set
\begin{equation}\label{eqn:Omega2}
  \mOmega_2
    =
    \sum_{i=1}^{M}
      \mV_{i2} \mW_i \mV_{i2}^\trans
      \otimes
      \mV_{i2} \mW_i \mV_{i2}^\trans,
\end{equation}
where $\otimes$ denotes Kronecker product (with the notational convention that
$\otimes$ has lower precedence than matrix multiplication).  When $\mOmega_2$
is invertible on the subspace corresponding to symmetric matrices,
define symmetric $q \times q$ matrix-valued function $\mhS(\vb)$ and
symmetric $q \times q$ matrix $\mB$ via the relation
\begin{align*}
  \vecm\{\mhS(\vb)\} &= \mOmega_2^{-1} \vecm\{\mhA(\vb)\}, \\
  \vecm(\mB) &=
    \mOmega_2^{-1}
    \vecm\big\{
      \sum_{i=1}^{M}
        \mV_{i2} \mW_i \mD_{i}^{-2} \mW_i^\trans \mV_{i2}^\trans
    \big\},
\end{align*}
where $\vecm(\cdot)$ denotes column vector concatenation.  
For all matrices $\mB$, $\mC$, and $\mX$ of consistent dimensions,
\(
  \vecm(\mB \mX \mC) = (\mC^\trans \otimes \mB) \vecm(\vX).
\)
It follows that
\[
  \E\{\mhS(\vbeta)\} =
    \mSigma
    +
    \phi \mB.
\]
In light of this relation,
define moment-based covariance matrix estimator $\mhSigma_{\mW}$ as
\begin{equation}\label{eqn:Sigma-hat}
  \mhSigma_{\mW}
  =
  \mhS(\vhbeta_{\mW})
  -
  \hat \phi \mB.
\end{equation}
Due to the dependence between $\vhbeta_{\mW}$ and $\vheta_i$, the matrix
$\mhSigma_{\mW}$ is not an unbiased estimate of $\mSigma$, but we will later
show that its bias is often negligible.

In practice, the estimate $\mhSigma_{\mW}$ may not be positive semidefinite.
To handle this situation, we can replace $\mhSigma_{\mW}$, by
$\mtSigma_{\mW}$,  the projection of $\mhSigma_{\mW}$ onto the cone of
positive semidefinite matrices.  \citet{Cart86} employ a similar modification.
For any continuous function, $g$, if the convergence
$\mhSigma_{\mW} \toP \mSigma$ holds, then
$g(\mhSigma_{\mW}) \toP g(\mSigma)$.  Thus, since projection onto the cone of
positive semidefinite matrices is a continuous function, by the continuous
mapping theorem,
if $\mhSigma_{\mW}$ is a consistent estimator of $\mSigma$, then
$\mtSigma_{\mW}$ is as well.

The estimator $\vhbeta_{\mW}$ as defined here is similar to the estimator used
by \citet{Swam70} and the other authors mentioned in
Section~\ref{sec:moment-based-est}, but, unlike the existing approaches, the
form in~\eqref{eqn:beta-hat} allows for rank-degenerate predictor matrices.
The estimator $\mhSigma_{\mW}$ is unique; earlier approaches used a simple
unweighted covariance estimate, which requires full-rank predictor
matrices to guarantee consistency.

\subsection{Application to hierarchical generalized linear models}
\label{sec:hglm}

For a hierarchical generalized linear model, we will require
subpopulation-specific effect estimators $\vheta_i$ for $\veta_i = [
\vbeta^\trans\ \vu_i^\trans]^\trans \ (i = 1, \dotsc, M)$ and a dispersion
estimator $\hat \phi$.  With these, we will apply the moment-based estimation
procedure described in the previous section to get estimators for $\vbeta$ and
$\mSigma$.

For most nonlinear models, the moment
relations~\eqref{eqn:conditional-moment-relations} will not hold exactly.
These relations will be approximations, with the quality of the approximation
depending on the relative sizes of $n_i$ and $p + q$.  When using the
moment-based procedure to estimate the parameters of a hierarchical
generalized linear model, the estimators $\vhbeta_{\mW}$ and $\mhSigma_{\mW}$
will be biased, and we will not be able to get theoretical performance
guarantees.  However, as we later demonstrate in
Sections~\ref{sec:hglm-simulation} and~\ref{sec:application}, in many
large-sample regimes, the moment
relations~\eqref{eqn:conditional-moment-relations} are reasonable
approximations, and the moment-based estimators perform well.

As in the linear case, some of the group-specific feature matrices $\mF_i = [
\mX_i\ \mZ_i ]\  (i = 1, \dotsc, M)$ may be rank-degenerate.  We can handle
these degeneracies by imposing linear identifiability constraints on the
group-specific estimates.  Specifically, letting $\mV_i$ be a matrix with
$r_i$ orthonormal columns spanning the row space of $\mF_i$, we will require
that $\vheta_i$ lie in the span of $\mV_i$.  With this constraint, under
standard regularity conditions, if the maximum likelihood estimator exists
then it will be unique, with conditional expectation $\E(\mV_i^\trans \vheta_i
\mid \vu_i) = \mV_i^\trans \veta_i + \oh(n_i^{-1/2})$ and conditional
covariance $\cov(\mV_i^\trans \vheta_i \mid \vu_i) = \phi \mV_i^\trans (\mF_i
\mLambda_i \mF_i)^{\pinv} \mV_i + \oh(n_i^{-1})$ for a matrix $\mLambda_i$
depending on $\vbeta$ and $\vu_i$.  We will use a plug-in estimate for
$\mLambda_i$, which will lead to a consistent estimate for $\cov(\mV_i^\trans
\vheta_i \mid \vu_i)$ as $n_i$ increases.

Unfortunately, even with the rank-degeneracy issue solved, the group-specific
maximum likelihood effect estimator may not exist for all $i$.  In logistic
regression models, this happens when the outcomes are perfectly separated by a
linear combination of the predictors.  One popular solution to this separation
problem is to modify the maximum likelihood estimator~\citep{Hein02}.
In particular, Firth's modified estimator and generalizations thereof are
particularly effective \citep{Firt93,Kosm09}; when the predictor matrix is of
full rank, not only do these estimators always exist, they reduce the bias
from $\oh(n_i^{-1/2})$ to $\oh(n_i^{-1})$.  In light of these properties, we
take $\vheta_i$ to be Firth's modified estimator instead of the maximum
likelihood estimator.

For $\hat \phi$, we will use a weighted combination of group-specific
dispersion estimates $\hat \phi_1, \dotsc, \hat \phi_M$.  With the usual
Pearson residual-based dispersion estimate, $\hat \phi_i$ will be approximately
distributed as a chi-squared random variable with $(n_i - r_i)$ degrees of
freedom, scaled by $\phi / (n_i - r_i)$.

The full procedure for estimating the parameters of a hierarchical generalized
linear model is as follows:
\begin{enumerate}
\item \label{step:group-effect-est} For each group $i = 1, \dotsc, M$:

\begin{enumerate}
\item Construct group-specific feature matrix $\mF_i = [\mX_i\ \mZ_i]$;
use a singular value decomposition to decompose this matrix as $\mF_i =
\mF_{0i} \mV_i^\trans$, where $\mF_{0i}$ has full column rank $r_i$ and
$\mV_i = [\mV_{i1}^\trans\ \mV_{i2}^\trans]^\trans$ is
a matrix of dimension $(p + q) \times r_i$ with orthonormal columns.

\item \label{step:firth} Use Firth's modified score function with data $(\vy_i,
\mF_{0i})$ to get group-specific effect estimate $\vheta_{0i}$.

\item Set $\mD_i^{2}$ to be a plug-in estimate of the unscaled
conditional precision matrix of $\vheta_{0i}$; that is, set $\mD_i^{-2}$ to
be a plug-in estimate of $\phi^{-1} \cov(\vheta_{0i} \mid \vu_i)$.

\item Set $\vheta_i = \mV_i \vheta_{0i}$.

\item If $\phi$ is unknown, compute group-specific dispersion estimate $\hat \phi_i$.
\end{enumerate}

\item 
\label{step:dispersion-est}
  If $\phi$ is unknown, compute pooled dispersion estimate
\[
  \hat \phi
  =
  \frac{
    \sum_{i=1}^{M} (n_i - r_i) \, \hat \phi_i
  }{
    \sum_{i=1}^{M} (n_i - r_i)
  };
\]
otherwise, set $\hat \phi = \phi$.

\item \label{step:weight-est}
Choose positive-definite weight matrices $\mW_1, \dotsc, \mW_M$.  With
these weights, use~\eqref{eqn:beta-hat} and~\eqref{eqn:Sigma-hat}
to compute estimates $\vhbeta = \vhbeta_{\mW}$ and $\mhSigma = \mhSigma_{\mW}$.

\item \label{step:project}
Check if $\mhSigma$ is positive semidefinite.  If not replace
$\mhSigma$ with a projection onto the positive semidefinite cone.

\item Optionally, use the esitmates $\vhbeta$ and $\mhSigma$ to
choose a new set of weight matrices and redo steps~(\ref{step:weight-est})
and~(\ref{step:project}).

\item If required, use normal approximations for the distributions of $\vu_i$
and $\vheta_{i} \mid \vu_i$ to compute empirical Bayes posterior mean and covariance
estimates for $\vu_i$:
\begin{align*}
  \widehat{\E}(\vu_i \mid \vy)
    &=
      \mC_i \mV_{i2}
      (\mV_{i}^\trans \vheta_i - \mV_{i1}^\trans \vhbeta), \\
  \widehat{\cov}(\vu_i \mid \vy)
    &=
    \hat \phi \, \mC_i,
\end{align*}
where
\(
    \mC_i
    =
    \mhSigma^{1/2}
    (
      \hat \phi \mI_q
      +
      \mhSigma^{1/2}
      \mV_{i2}^\trans \mD_{i}^2 \mV_{i2}
      \mhSigma^{1/2}
    )^{-1}
    \mhSigma^{1/2}.
\)
These quantities exist even if $\mhSigma$ does not have full rank.
\end{enumerate}

If we assume that at most a constant number of iterations are required in
step~(ii), then
the computational complexity for fitting the $i$th group in
step~(\ref{step:group-effect-est}) is of order $\Oh(n_i r_i^2)$, so that
the total cost of step~(\ref{step:group-effect-est}) is of order $\Oh\{N (p + q)^2\}$.
Step~(\ref{step:dispersion-est}) has cost $\Oh(M)$.
For all choices of weight matrices discussed in this report, computing $\mW_i$
requires at most $\Oh\{r_i q (r_i + q)^2 + r_i^3\}$ operations, so that
computing all $M$ weight matrices has cost $\Oh\{M (p + q)^3\}$.
Once the weights have been
computed, it takes $\Oh\{M p (p + q)^2\}$ operations to compute $\vhbeta_{\mW}$,
followed by $\Oh(M q (p + q)^2)$ to compute $\mhA(\vhbeta_{\mW})$ and
$\Oh(M p q^2 + M q^4)$ to compute $\mOmega_2$.
These are the dominant consts.  Conservatively, step~(\ref{step:weight-est})
requires $\Oh\{ M (p + q)^3 + M q^4 \}$ operations.
Step~(\ref{step:project}) has cost $\Oh(q^3)$.  The costs for the remaning
steps are similar to those already discussed.

In total, at most $O\{N (p + q)^2 + M (p + q)^3 + M q^4\}$ operations are
required.  This bound uses the approximation $r_i = \Oh(p + q)$, which is
often conservative.  In fact, in situations where the column space of $\mZ_i$
is contained in the column space of $\mX_i$ for all $i$, we will have
$r_i \leq p$.  In this scenario, at most $\Oh(N p^2 + M p^3 + M q^4)$ operations
are required.

Notably, once the group-specific effect estimate $\vheta_i$, the conditional
precision estimate $\mD_i^{2}$, and the dispersion estimate $\hat \phi_i$ have
been computed, the procedure has no need for $\vy_i$ and $\mF_i$.  This is
both a strength and a weakness.  It is a strength because it reduces the
computation and the memory demands of the procedure, and it allows most of the
operations to be trivially parallelized.  The weakness in this data reduction
is that it likely sacrifices statistical efficiency.  On balance, as later we
demonstrate in Sections~\ref{sec:simulations} and~\ref{sec:application},
in many large-scale data regimes it is worthwhile to make this
computational-statistical trade-off.

\section{Weight choices}
\label{sec:weight-choices}

\subsection{Weighted, unweighted, and semi-weighted cases}

The estimators introduced in Section~\ref{sec:estimation-procedure} depend on
a choices of weights $\mW_i\ (i = 1, \dotsc, M)$.  The choice that minimizes
$\E\norm{\vhbeta_{\mW} - \vbeta}^2$ is
\begin{equation}\label{eqn:optimal-weight}
  \mW_i
    =
    (\mV_{i2}^\trans \mbSigma \mV_{i2} + \mD_i^{-2})^{-1},
\end{equation}
where $\mbSigma = \phi^{-1} \mSigma$.  In general, we do not know $\mSigma$
and $\phi$, so we cannot use these weights.

In the univariate case, Cochran discusses three practical alternatives.  The
first option, which he calls the ``unweighted'' method, corresponds to setting
\(
  \mW_i = \mI_{r_i}.
\)
The second option, which Cochran calls ``weighted,'' corresponds to setting
\(
  \mW_i = \mD_i^2.
\)
The last option depends on an initial choice $\mbSigma_0$ and corresponds to
setting
\(
  \mW_i
    =
    (\mV_{i2}^\trans \mbSigma_0 \mV_{i2} + \mD_i^{-2})^{-1};
\)
Cochran calls this the ``semi-weighted'' method.  Following Cochran and Swamy,
we use a two-step estimation scheme, taking an initial choice of weights to
get a preliminary estimate $\mhSigma_0$ of the scaled random effect covariance
matrix, and then using this estimate with the semi-weighted method to choose a
new set of weights, repeating the estimation process.
For the initial choice of weights, we use the semi-weighed method with
$\mbSigma_0$ chosen as specified in the following section.


\subsection{Optimal weights}

In this section, we will study the optimal weight choice. We do not give a
complete analysis, but we will derive a heuristic choice based on minimax
optimality considerations.  We will show that, after standardizing the
predictors, it is reasonable (and sometimes optimal) to choose the
semi-weighted $\mW_i$ with $\mbSigma_0 = \mI_q$.

For $i = 1, \dotsc, M$, set $\vhtheta_i = \mV_{i}^\trans \vheta_i$.  We will use a
weighted combination of the estimators $\vhtheta_1, \dotsc, \vhtheta_M$ to
estimate $\vbeta$.
Let $\malpha = (\malpha_1, \dotsc, \malpha_M)$ be a vector of weight matrices,
where component matrix
$\malpha_i$ has size $r_i \times p$.  Define estimator
\(
  \vhbeta_{\malpha} = \sum_{i=1}^{M} \malpha_i^\trans \vhtheta_i,
\)
which has expectation
\(
  \E(\vhbeta_{\malpha})
  =
  \big(\sum_{i=1}^{M} \mV_{i1} \malpha_i\big)^\trans
  \vbeta
\)
and covariance
\(
  \cov(\vhbeta_{\malpha})
  =
  \sum_{i=1}^{M}
    \malpha_i^\trans \{ \cov(\vhtheta_i) \} \malpha_i.
\)
For $\vhbeta_{\malpha}$ to be unbiased for all $\vbeta$, we must have
\(
  \sum_{i=1}^{M} \mV_{i1} \malpha_i = \mI_p.
\)

Among all choices of $\malpha$ that make $\vhbeta_{\malpha}$ unbiased,
the one that minimizes
the mean squared error $\E\|\vhbeta_{\malpha} - \vbeta\|_2^2$ is the one
minimizing $\tr\{\cov(\vhbeta_{\malpha})\}$.  Letting $\valpha_{ik}$ denote
the $k$th column of $\malpha_i$, the squared-error-optimal choice of $\malpha$
must satisfy the Lagrangian gradient equations
\[
    \{ \cov(\vhtheta_i) \} \valpha_{ik} = \mV_{i1}^\trans \vomega_k,
\]
with $p \times p$ Lagrange multiplier matrix
$\mOmega = [ \vomega_1 \cdots \vomega_p]$.  Thus, the optimal unbiased weight vector
satisfies
\[
  \malpha^\ast_i = \{ \cov(\vhtheta_i) \}^{-1} \mV_{i1}^\trans \mOmega,
\]
with
\(
  \mOmega
  =
  \big[
    \sum_{i=1}^{M}
      \mV_{i1} \{ \cov(\vhtheta_i) \}^{-1} \mV_{i1}^\trans
  \big]^{-1};
\)
minimizing estimator $\vhbeta^\ast$ has
\(
  \cov(\vhbeta^\ast) = \mOmega.
\)

The weight $\malpha^\ast$ depends on the unknown quantity
$\mbSigma = \phi^{-1} \mSigma$.  We would like to find a weight which is independent
of these unknowns.  To measure the sub-optimality of any particular choice of
$\valpha$, assume $\phi = 1$ without loss of generality, and
define the risk function
\[
  R(\mSigma, \malpha)
  = \tr\{ \mOmega^{-1} \cov(\vhbeta_{\malpha}) \}.
\]
Ideally, we should choose the weights that minimize the maximum
risk.  In practice, it is difficult to solve the underlying optimization
problem to find this set of values for $\malpha$, so we instead will
choose the weights $\malpha$ based on a heuristic.

Define extremal risks $R_0(\malpha)$ and $R_{\infty}(\malpha)$ as
\begin{align*}
  R_0(\malpha)
    &= \lim_{t \to 0} R(t \mI_q, \malpha)
    = \tr\Big[
        \Big\{ \sum_{i=1}^{M} \mV_{i1} \mD_i^2 \mV_{i1}^\trans \Big\}
        \Big\{ \sum_{i=1}^{M} \malpha_{i}^\trans \mD_i^{-2} \malpha_i \Big\}
      \Big],
\\
  R_{\infty}(\malpha)
    &= \lim_{t \to \infty} R(t \mI_q, \malpha)
    = \tr\Big[
        \Big\{ \sum_{i=1}^{M}
                 \mV_{i1} (\mV_{i2}^\trans \mV_{i2})^{\pinv} \mV_{i1}^\trans
        \Big\}
        \Big\{ \sum_{i=1}^{M}
                 \malpha_{i}^\trans \mV_{i2}^\trans \mV_{i2} \malpha_i
        \Big\}
      \Big].
\end{align*}
Instead of finding $\malpha$ to minimize $\sup_{\mSigma} R(\mSigma, \malpha)$,
we will attempt to find weights that minimize the average
$\bar R(\malpha) = (R_0(\malpha) + R_\infty(\malpha))/2$.
To this end, set
\[
  \mB = \sum_{i=1}^{M} \mV_{i1} \mD_i^2 \mV_{i1}^\trans,
\qquad
  \mC = \sum_{i=1}^{M}
           \mV_{i1} (\mV_{i2}^\trans \mV_{i2})^{\pinv} \mV_{i1}^\trans.
\]
For $\mbalpha$ to minimize $\bar R$,
while simultaneously satisfying the unbiasedness constraint, its
$i$th component must satisfy the Lagrangian gradient equation
\[
  \mD_i^{-2} \mbalpha_i \mB
  +
  \mV_{i2}^\trans \mV_{i2} \mbalpha_i \mC
  =
  \mV_{i1}^\trans \mLambda
\]
for some $p \times p$ matrix of Lagrange multipliers, $\mLambda$,
independent of $i$.  In vector form,
\[
  (
  \mB^\trans \otimes \mD_i^{-2}
  +
  \mC^\trans \otimes \mV_{i2}^\trans \mV_{i2}
  )
  \vecm(\mbalpha_{i})
  =
  (\mI_p \otimes \mV_{i1}^\trans) \vecm(\mLambda).
\]
The unbiasedness constraint $\sum_{i=1}^{M} \mV_{i1} \mbalpha_i = \mI_p$ must
also hold.

Finding $\mbalpha$ and $\mLambda$ requires solving a linear system of
$p \sum_{i=1}^{M} r_i + p^2$ equations in as many unknowns.  For general
situations, this is computationally expensive.  However, in the case of a
hierarchical generalized linear models satisfying
$\sum_{i=1}^{M} \mX_i^\trans \mX_i = M \mI_p$ and $\mX_i = \mZ_i$ for all
$i$, we get the simplification $\mB = \mC = M \mI_p$; in this case, the optimal
weight is
\[
  \mbalpha_i = (\mV_{i2}^\trans \mV_{i2} + \mD_i^{-2})^{-1} \mV_{i1}^\trans
  \mbOmega,
\]
with $\mbOmega$ chosen such that $\sum_{i=1}^{M} \mbalpha_i^\trans \mV_{i1} =
\mI$.  This corresponds to the semi-weighted case using $\mbSigma_0 = \mI_q$.
Motivated by this correspondence, in practical applications we will
standardize the predictors and then use the semi-weights with $\mbSigma_0 =
\mI_q$.  In addition to the optimality considerations, the standardization
ensures that the procedure is equivariant.

\section{Finite sample properties of moment-based estimates}
\label{sec:finite-sample-prop}

\subsection{Theoretical framework}

To analyze the performance of the proposed moment-based estimation procedure,
we will need to be precise about what assumptions are required.  To facilitate
asymptotic analysis, we will state these assumptions in terms of sequences
indexed by $N$.  We make this dependence on $N$ explicit in the assumption
statements, but, to simplify the notation, will suppress this dependence
in most of the text.

\begin{assumption}\label{asn:effect}\label{asn:first}
There exists a non-random $p$-dimensional fixed effect vector $\vbeta$ and, for
each value of $N$ there is a sequence of $M(N)$ independent and identically
distributed $q$-dimensional random effect vectors:
$\vu_{N,1}, \dotsc, \vu_{N,M(N)}$.  The $i$th random effect vector can be
expressed as $\vu_{N,i} = \mSigma^{1/2} \vtu_{N,i}$ where $\mSigma^{1/2}$ is the
symmetric square root of positive semidefinite matrix $\mSigma$, and the
sphered random effect vector $\vtu_{N,i}$ satisfies the moment conditions
\begin{subequations}
\begin{align}
  \E(\vtu_{N,i}) &= 0, \\
  \cov(\vtu_{N,i}) &= \mI_q, \\
  \E\norm{\vtu_{N,i}}^4 &\leq \mu
\end{align}
\end{subequations}
for some finite constant $\mu$.
\end{assumption}

\begin{assumption}\label{asn:identifiable}
For each $N$ and
all $i = 1, \dotsc, M(N)$ there exists a matrix with orthonormal columns
$\mV_{N,i} = [ \mV_{N,i1}^\trans\ \mV_{N,i2}^\trans ]^\trans$,
and a symmetric
positive-definite matrix $\mD_{N,i}$ (not necessarily diagonal) such that
$\mV_{N,i1}$ and $\mV_{N,i2}$ have dimensions $p \times r_{N,i}$ and
$q \times r_{N,i}$, respectively, and $\mD_{N,i}$ has dimension
$r_{N,i} \times r_{N,i}$.  Further, the following conditions hold:
\begin{enumerate}

\item The matrix
\(
  \sum_{i=1}^{M(N)} \mV_{N,i1} \mV_{N,i1}^\trans
\)
is invertible.

\item The matrix
\(
  \sum_{i=1}^{M(N)}
    (\mV_{N,i2} \mV_{N,i2}^\trans)
    \otimes
    (\mV_{N,i2} \mV_{N,i2}^\trans)
\)
is invertible on the subspace $\sS_q$ of vectors $\vs$ satisfying
\(
  \vs = \vecm(\mS)
\)
for some symmetric $q \times q$ matrix $\mS$.
\end{enumerate}
\end{assumption}

\begin{assumption}\label{asn:effect-est}
Letting $\veta_{N,i} = [ \vbeta^\trans \ \vu_{N,i}^\trans ]^\trans$ be the true
$(p+q)$-dimensional effect vector for the $i$th group, there exist
group-specific effect estimates $\vheta_{N,1}, \dotsc, \vheta_{N,M(N)}$
such that the estimation error $\vh_{N,i} = \veta_{N,i} - \vheta_{N,i}$
satisfies the moment relations
\begin{subequations}
\begin{align}
  \E(\mV_{N,i}^\trans \vh_{N,i}) &= 0, \\
  \cov(\mV_{N,i}^\trans \vh_{N,i}) &= \phi \mD_{N,i}^{-2}, \\
  \E\norm{\phi^{-1/2} \mD_{N,i} \mV_{N,i}^\trans \vh_{N,i}}^4 &\leq
    \lambda
\end{align}
\end{subequations}
for some dispersion parameter $\phi$ and finite constant $\lambda$.
Furthermore, the estimation errors $\vh_{N,1}, \dotsc, \vh_{N,M(N)}$ and
the random effects $\vu_{N,1}, \dotsc, \vu_{N,M(N)}$ are mutually independent.
\end{assumption}

\begin{assumption}\label{asn:dispersion-est}
For each $N$ there exists a random dispersion parameter estimate
$\hat \phi_N$ independent of the vectors $\vh_{N,1}, \dotsc, \vh_{N,M(N)}$
and $\vu_{N,1}, \dotsc, \vu_{N,M(N)}$ such that
\begin{equation}
  \E(\hat \phi_N / \phi - 1)^2 \leq \nu/(N - \rho_N)
\end{equation}
where $\rho_N = \sum_{i=1}^{M(N)} r_{N,i} < N$ and $\nu < \infty$.
\end{assumption}

These assumptions are motivated by the linear case introduced in
Section~\ref{sec:intuition}.  Assumption~\ref{asn:identifiable}(a) ensures
that $\vbeta$ is identifiable; it holds if and only if the combined predictor
matrix $\mX = [ \mX_1^\trans\ \cdots\ \mX_M^\trans ]^\trans$ has full column
rank; Assumption~\ref{asn:identifiable}(b) ensures that $\mSigma$ is
identifiable; it holds if and only if
$\sum_{i=1}^{M} (\mZ_i^\trans \mZ_i) \otimes (\mZ_i^\trans \mZ_i)$
is invertible on $\sS_q$.  Assumption~\ref{asn:effect-est} holds for the
hierarchical linear model whenever $\E|\varepsilon_{ij}|^4 < \infty$; for
nonlinear models, including hierarchical generalized linear models,
Assumption~\ref{asn:effect-est} will not hold exactly, but it will be a
reasonable approximation whenever the group-specific sample sizes are large.
For Assumption~\ref{asn:dispersion-est}, in models where the dispersion
parameter is known it suffices to take $\hat \phi_N = \phi$ and $\nu = 0$.

\begin{assumption}\label{asn:weight}\label{asn:last}
For each $N$ there exists a sequence of
symmetric positive-definite weight matrices
$\mW_{N,1}, \dotsc, \mW_{N,M(N)}$ where the $i$th weight matrix has
dimension $r_{N,i} \times r_{N,i}$ and satisfies the relation
\begin{equation}
  \mW_{N,i}
  (\mV_{N,i2}^\trans \mSigma \mV_{N,i2} + \phi \mD_{N,i}^{-2})
  \mW_{N,i}
  \preceq
  \kappa_N \mW_{N,i}
\end{equation}
for some nonrandom sequence $\kappa_N$ independent of $i$.
\end{assumption}

Table~\ref{tab:weight-constants} shows the bounding constants from
Assumption~\ref{asn:weight} associated with each weight method
discussed in Section~\ref{sec:weight-choices}.  It is
straightforward to derive these bounds for the unweighted and weighted cases.
For the semi-weighted case, we derive the bound in
Lemma~\ref{lem:semiweight-bound}.  Generally, $\norm{\mD_i}$ will scale
proportionally to the square root of the group-specific sample size,
$n_i^{1/2}$.  We can see that the bound for the unweighted case degrades if
some $n_i$ is small, while the bound for the weighted case degrades if some
$n_i$ is large.  The bound for the semi-weighted case is insensitive to the
group-specific sample sizes.

\begin{lemma}\label{lem:semiweight-bound}
If $\mbSigma_0$ and $\mD_1, \dotsc, \mD_M$ are positive-definite and
$\mSigma$ is positive-semidefinite, then for the weight defined by
\(
  \mW_i = (\mV_{i2}^\trans \mbSigma_0 \mV_{i2} + \mD_i^{-2})^{-1},
\)
Assumption~\ref{asn:weight} holds with
\(
  \kappa = \|\mbSigma_0^{-1} \mSigma \| + \phi.
\)
\end{lemma}
\begin{proof}
We will drop the subscript $i$ for the proof of the lemma.  First,
note the relation
\(
  \mD^{-1} \mW \mD^{-1}
  = (\mD \mV_2^\trans \mbSigma_0 \mV_2 \mD + \mI_r)^{-1}
  \preceq \mI_r,
\)
so that
\begin{equation}\label{eqn:WDW-bound}
  \phi \mW^{1/2} \mD^{-2} \mW^{1/2} \preceq \phi \mI_r.
\end{equation}
Next, use the matrix inversion lemma to express
\[
  \mW
    =
    \mD^2
    -
    \mD^2 \mV_2^\trans
    (\mbSigma_0^{-1} + \mV_2 \mD^2 \mV_2^{\trans})^{-1}
    \mV_2 \mD^2.
\]
Use the identities $\mI - (\mA + \mB)^{-1} \mB = (\mA + \mB)^{-1} \mA$
and $\mB (\mA + \mB)^{-1} = \mI - \mA (\mA + \mB)^{-1}$ to get
\[
  \mV_2 \mW \mV_2^\trans
  =
  \mbSigma_0^{-1/2}
  \{
    \mI_q
    -
    (\mI_q + \mbSigma_0^{1/2} \mV_2 \mD^2 \mV_2^\trans \mbSigma_0^{1/2})^{-1}
  \}
  \mbSigma_0^{-1/2}.
\]
Employing the bound $\mI - (\mI + \mA)^{-1} \preceq \mI$, which holds for any
positive-semidefinite matrix $\mA$, it follows that
\(
  \mV_2 \mW \mV_2^\trans
  \preceq
  \mbSigma_0^{-1}.
\)
Thus,
\begin{equation}\label{eqn:WSigmaW-bound}
  \mW^{1/2} \mV_2^\trans \mSigma \mV_2 \mW^{1/2}
  \preceq
  \|\mSigma \mbSigma_0^{-1}\| \mI_r.
\end{equation}
The result of the lemma follows from~\eqref{eqn:WDW-bound}
and~\eqref{eqn:WSigmaW-bound}.
\end{proof}

\begin{table}
\caption{Weight choices and associated bounding constants}
\label{tab:weight-constants}

\centering

\fbox{%
\begin{tabular}{lcc}
Method
& $\mW_{i}$
& $\kappa$
\\
\hline
Unweighted
& $\phantom{{}_{r_i}}\mI_{r_i}$
& $\norm{\mSigma} + \phi \max_{i} \norm{\mD_i^{-2}}$
\\
Weighted
& $\phantom{{}_i^{2}}\mD_i^{2}$
& $\norm{\mSigma} \max_{i} \norm{\mD_i^2} + \phi$
\\
Semi-Weighted
& $\phantom{{}^{-1}}(\mV_{i2}^\trans \mbSigma_0 \mV_{i2} + \mD_{i}^{-2})^{-1}$
& $\norm{\mbSigma_0^{-1} \mSigma} + \phi$
\\
\end{tabular}}
\end{table}

\subsection{Existence}

For the estimates $\vhbeta_{\mW}$ and $\mhSigma_{\mhSigma}$ to be
well-defined, we must have that the corresponding quantities $\mOmega$ and
$\mOmega_2$ are invertible.  Propositions~\ref{prop:Omega-invertible}
and~\ref{prop:Omega2-invertible} show that this is always the case whenever
the group-specific weights are positive definite and
Assumption~\ref{asn:identifiable} is in force.

\begin{proposition}\label{prop:Omega-invertible}
For $i = 1, \dotsc, M$ let $\mW_i$ be a nonrandom symmetric positive-definite
matrix.  If Assumption~\ref{asn:identifiable}(a) holds, then the matrix
\(
  \mOmega
\)
defined in~\eqref{eqn:Omega1} is invertible,
so that $\vhbeta_{\mW}$ is well-defined.
\end{proposition}
\begin{proof}
The matrix $\mOmega$ is symmetric, so it suffices to show that it is
positive-definite.  We will proceed by contradiction.  Suppose that the
statement of the proposition is false, so that for some nonzero vector
$\vt$, the identity $\vt^\trans \mOmega \vt = 0$
holds.  In this case, since $\mW_i$ is positive-definite, it must follow that
$\mV_{i1}^\trans \vt = 0$ for all $i = 1, \dotsc, M$.  Thus,
\(
  \sum_{i=1}^{M} \vt^\trans \mV_{i1} \mD_i^2 \mV_{i1}^\trans \vt = 0.
\)
This contradicts Assumption~\ref{asn:identifiable}(a).  It must follow, then,
that $\vt^\trans \mOmega \vt > 0$ for all nonzero $\vt$, so that
$\mOmega$ has full rank.
\end{proof}

We state the result for $\mOmega_2$, which follows by a similar argument,
as Proposition~\ref{prop:Omega2-invertible}.
The full proof of this result is given in
Appendix~\ref{sec:Omega2-invertible-proof} of the on-line supplement.

\begin{proposition}\label{prop:Omega2-invertible}
For $i = 1, \dotsc, M$ let $\mW_i$ be a nonrandom symmetric positive-definite
matrix.  If Assumption~\ref{asn:identifiable}(b) holds, then the matrix
\(
  \mOmega_2
\)
defined in~\eqref{eqn:Omega2}
is invertible on $\sS_q$, so that $\mhSigma_{\mW}$ is well-defined.
\end{proposition}

\subsection{Concentration}

The next results, Corollary~\ref{cor:beta-concentration} and
Proposition~\ref{prop:sigma-concentration}, show that with high probability,
$\vhbeta_{\mW}$ and $\mhSigma_{\mW}$ are close to their estimands.

\begin{proposition}\label{prop:beta-est-moment}
If Assumptions~\ref{asn:effect},~\ref{asn:identifiable},~\ref{asn:effect-est},%
and~\ref{asn:weight} are in force,
then $\vhbeta_{\mW}$ satisfies the moment relations
\begin{subequations}
\label{eqn:beta-est-moments}
\begin{align}
  \label{eqn:beta-est-unbiased}
  \E(\vhbeta_{\mW}) &= \vbeta, \\
  \label{eqn:beta-est-cov-bound}
  \cov(\vhbeta_{\mW}) &\preceq \kappa \mOmega^{-1}.
\end{align}
\end{subequations}
\end{proposition}
\begin{proof}
We have
\(
  \vhbeta_{\mW}
    =
      \mOmega^{-1}
      \sum_{i=1}^{M}
        \mV_{i1} \mW_{i} \mV_{i}^\trans \vheta_i.
\)
Proposition~\ref{prop:Omega-invertible} shows if Assumption~\ref{asn:identifiable} is
in force, then $\mOmega$ is invertible and consequently $\vhbeta_{\mW}$ is
well-defined.  Assumptions~\ref{asn:effect} and~\ref{asn:effect-est} imply that
\(
  \E(\mV_i^\trans \vheta_i) = \mV_{i1}^\trans \vbeta,
\)
so that $\E(\vhbeta_{\mW}) = \vbeta$.  Additionally, these assumptions
together with Assumption~\ref{asn:weight} imply that
\[
  \cov(\vhbeta_{\mW})
   =
      \mOmega^{-1}
      \big\{
        \sum_{i=1}^{M}
        \mV_{i1} \mW_i
        (\mV_{i2}^\trans \mSigma \mV_{i2} + \phi \mD_{i}^{-2})
        \mW_i \mV_{i1}^\trans
      \big\}
      \mOmega^{-1}
    \preceq
      \kappa \mOmega^{-1}. \qedhere
\]
\end{proof}

\begin{corollary}\label{cor:beta-concentration}
If Assumptions~\ref{asn:effect},~\ref{asn:identifiable},~\ref{asn:effect-est},%
and~\ref{asn:weight} are in force, then for any $\varepsilon > 0$,
\[
  \Pr\{
    \norm{\vhbeta_{\mW} - \vbeta}^2
    \geq
    \varepsilon^{-1}
    \kappa \tr(\mOmega^{-1})
  \}
    \leq \varepsilon.
\]
\end{corollary}
\begin{proof}
From Proposition~\ref{prop:Omega-invertible} it follows that
\begin{align*}
  \E\norm{\vhbeta_{\mW} - \vbeta}^2
    &= \E[\tr\{(\vhbeta_{\mW} - \vbeta)(\vhbeta_{\mW} - \vbeta)^\trans\}] \\
    &= \tr\{\cov(\vhbeta_{\mW})\} \\
    &\leq \kappa \tr(\mOmega^{-1}).
\end{align*}
Now apply Markov's inequality.
\end{proof}

\begin{proposition}\label{prop:sigma-concentration}
If Assumptions~\ref{asn:effect}--\ref{asn:weight} are in force, then for any
$\varepsilon \in (0,1]$,
\[
  \Pr\{\norm{\mhSigma_{\mW} - \mSigma}_{\Frob}^2 \geq \varepsilon^{-2}
    \kappa^2 C^2 \tr(\mOmega_{2}^{-1})
  \}
  \leq
  \varepsilon,
\]
where
\(
  C
  =
    \{
      9 p^{3/2}
      + 3(\lambda + 2)^{1/2}
      + \mu^{1/2}
      + \nu^{1/2} (N/\rho - 1)^{-1/2}
    \}/2.
\)
\end{proposition}
\begin{proof}
Define $\mS$ analogously to $\mhS$ be replacing $\vheta_i$ with $\veta_i$.
The triangle inequality implies that
\[
  \norm{\mhSigma_{\mW} - \mSigma}_{\Frob}
    \leq
    \norm{\mhS(\vhbeta_{\mW}) - \mhS(\vbeta)}_{\Frob}
    +
    \norm{\mhS(\vbeta) - \mS(\vbeta) - \phi \mB}_{\Frob}
    +
    \norm{\mS(\vbeta) - \mSigma}_{\Frob}
    +
    \abs{\phi - \hat \phi} \norm{\mB}_{\Frob}
\]
We analyze the right hand side summands in
Appendix~\ref{sec:sigma-concentration-support} of the on-line supplement;
Lemma~\ref{lem:sigma-concentration-supp}, stated after the proof of
Prop.~\ref{prop:sigma-concentration}, summarizes these results.

Fix any $a > 0$.  Set $\omega = \tr(\mOmega_2^{-1})$.
Lemma~\ref{lem:sigma-concentration-supp}(\ref{lem:item:shatbhat-shatb})
shows that
\[
  \Pr(\norm{\mhS(\vhbeta_{\mW}) - \mhS(\vbeta)}_{\Frob}
    \geq 9 a p^{3/2} \kappa \omega^{1/2}
  )
  \leq a^{-1}.
\]
Lemma~\ref{lem:sigma-concentration-supp}(\ref{lem:item:shatb-sb})
and Markov's inequality imply that
\[
  \Pr\{
    \norm{\mhS(\vbeta) - \mS(\vbeta) - \phi \mB}_{\Frob}
    \geq 3 a (\lambda + 2)^{1/2} \kappa \omega^{1/2}
  \}
  \leq a^{-2}.
\]
Similarly,
Lemma~\ref{lem:sigma-concentration-supp}(\ref{lem:item:sb-sigma})
and Markov's inequality imply that
\[
  \Pr\{
    \norm{\mS(\vbeta) - \mSigma}_{\Frob}
    \geq a \mu^{1/2} \kappa \omega^{1/2}
  \}
  \leq a^{-2}.
\]
For the final term,
Assumption~\ref{asn:dispersion-est} implies that
\[
  \Pr\{
    \abs{\hat \phi / \phi - 1}
    \geq a \nu^{1/2}(N - \rho)^{-1/2}
  \}
  \leq 1/a^2,
\]
and
Lemma~\ref{lem:sigma-concentration-supp}(\ref{lem:item:bias})
implies that
\[
  \phi \norm{\mB}_{\Frob}
    \leq \kappa \rho^{1/2} \norm{\mOmega_2^{-1}}
    \leq \kappa \rho^{1/2} \omega^{1/2}.
\]
Thus, with probability at least $1 - (1/a + 3/a^2)$,
\[
  \norm{\mhSigma_{\mW} - \mSigma}_{\Frob}
    <
    a \kappa \omega^{1/2}
    \{
      9 p^{3/2}
      + 3(\lambda + 2)^{1/2}
      + \mu^{1/2}
      + \nu^{1/2} (N/\rho - 1)^{-1/2}
    \}.
\]
Set $\varepsilon = (1/a + 3/a^2)$.  If $\varepsilon \leq 1$, then
\(
  a^{-1} = \sqrt{1 + 12 \varepsilon} - 1 > 2 \varepsilon.
\)
This gives the desired result.
\end{proof}

\begin{lemma}\label{lem:sigma-concentration-supp}
If Assumptions~\ref{asn:effect},~\ref{asn:identifiable},~\ref{asn:effect-est},%
and~\ref{asn:weight} are in force, then the following identities hold:
\begin{enumerate}
  \item \label{lem:item:shatbhat-shatb} \label{item:sigma-concentration-supp:first}
    \(
  \Pr\{ \| \mhS(\vhbeta_{\mW}) - \mhS(\vbeta) \|_{\Frob}
    <
    9 p^{3/2} \kappa
  \{ \tr(\mOmega_2^{-1}) \}^{1/2} / \varepsilon\} \geq 1 - \varepsilon,
  \)
\item \label{lem:item:shatb-sb}
  \(
  \E\| \mhS(\vbeta) - \mS(\vbeta) - \phi \mB \|_{\Frob}^2
  \leq 9 \kappa^2 (\lambda + 2) \tr(\mOmega_2^{-1}),
  \)
\item \label{lem:item:sb-sigma}
  \(
  \E\norm{\mS(\vbeta) - \mSigma}_{\Frob}^2
  \leq
      \mu \kappa^2 \tr(\mOmega_2^{-1}),
      \)
\item \label{lem:item:bias} \label{item:sigma-concentration-supp:last}
  \(
  \norm{\mB}_{\Frob}
  \leq
      \phi^{-1}
      \kappa
      \rho^{1/2}
      \norm{\mOmega_2^{-1}}^{1/2},
\)
\end{enumerate}
where $\rho = \sum_{i=1}^{M} r_i$.
\end{lemma}

\subsection{Near relative efficiency}
\label{sec:near-rel-efficiency}

We now show that with the semi-weighted method, if the initial choice for
$\mbSigma_0$ is close to the true value $\mbSigma = \phi^{-1} \mSigma$, then
the weighted estimate is close to optimal unbiased weighted estimate.  In this
sense, it is close to being ``relatively efficient''.

To be precise about this equivalence in efficiency, let $\vtheta_0$ denote the
vector with $1 + q \, (q + 1) / 2$ components, gotten by concatenating
$\phi$ and the unique elements of $\mSigma$.  For any parameter vector
$\vtheta$ with the same dimension, let
$\mSigma_{\vtheta}$ and $\phi_{\vtheta}$ denote the corresponding values of
the random effect covariance matrix and the dispersion parameter.
Set
\begin{equation}\label{eqn:beta-omega}
  \vhbeta_{\vtheta}
    =
    \mOmega_{\vtheta}^{-1}
    \sum_{i=1}^{M}
    \mV_{i1} \mW_{\vtheta i} \mV_{i}^\trans \vheta_i,
\end{equation}
where
\begin{equation}\label{eqn:wi-omega}
  \mW_{\vtheta i}
  =
  (\mV_{i2}^\trans \mbSigma_{\vtheta} \mV_{i2} + \mD_{i}^{-2})^{-1},
\quad
\mOmega_{\vtheta} = \sum_{i=1}^{M} \mW_{\vtheta i},
\end{equation}
and $\mbSigma_{\vtheta} = \phi_{\vtheta}^{-1} \mSigma_{\vtheta}$.  
If Assumption~\ref{asn:identifiable} is in force, then Proposition~\ref{prop:Omega-invertible}
implies that $\mOmega_{\vtheta}^{-1}$ exists and $\vhbeta_{\vtheta}$ exists
for all $\vtheta$.
Define $\vhbeta$, $\mW_i$, and $\mOmega$ as the quantities gotten by setting
$\vtheta = \vtheta_0$.

The next result states that for all parameter vectors $\vtheta$ in a
neighbourhood of $\vtheta_0$, the estimate $\vhbeta_{\vtheta}$ is uniformly
close to $\vhbeta$.  \citet{Cart86} state a similar asymptotic result in the
context of Swamy's estimation procedure; their heuristic proof of this result uses
different but related techniques.

\begin{proposition}\label{prop:finite-rel-efficiency}
Let $\sB$ be any neighbourhood of the true parameter vector $\vtheta_0$.
For any $\varepsilon > 0$, if Assumptions~\ref{asn:effect}--\ref{asn:effect-est} are
in force, then
\[
  \Pr\{
    \sup_{\vtheta \in \sB} \norm{\mOmega^{1/2} (\vhbeta_{\vtheta} - \vhbeta)}^2
    \geq
    C
    \varepsilon^{-1}
    \tau^2
  \}
    \leq
  \varepsilon,
\]
where
\(
  \tau =
  \sup_{\vtheta \in \sB} \max\{
    \norm{\mbSigma^{-1} \mbSigma_{\vtheta} - \mI_q},
    \norm{\mbSigma_{\vtheta}^{-1} \mbSigma - \mI_q}
  \}
\)
and
\[
  C = 9 \{
    48 p^3
    +
    48 p q^3 (1 + 4 p \tau)^2
    +
    768 \rho \tau^2 (1 + 4 p \tau)
  \},
\]
with
\(
  \rho = \sum_{i=1}^{M} r_i.
\)
\end{proposition}

\begin{proof}
For any vector $\vtheta$, write
\(
  \vhbeta_{\vtheta} - \vbeta
    =
      \mOmega_{\vtheta}^{-1}
      \sum_{i=1}^{M}
        \mV_{i1}
        \mW_{\vtheta i}
        (\mV_{i}^\trans \vheta_i - \mV_{i1}^\trans \vbeta).
\)
Now,
\[
  \vhbeta_{\vtheta} - \vhbeta
  =
    (\vhbeta_{\vtheta} - \vbeta)
    -
    (\vhbeta - \vbeta)
  =
    \sum_{i=1}^{M}
      (
        \mOmega_{\vtheta}^{-1} \mV_{i1} \mW_{\vtheta i}
        -
        \mOmega^{-1} \mV_{i1} \mW_{i}
      )
      (\mV_{i}^\trans \vheta_i - \mV_{i1}^\trans \vbeta).
\]
Set
\(
  \vgamma_i = \mW_i^{1/2} (\mV_{i}^\trans \vheta_i - \mV_{i1}^\trans \vbeta).
\)
It follows
that
\[
  \vhbeta_{\vtheta} - \vhbeta
  =
    (
      \mOmega_{\vtheta}^{-1}
      -
      \mOmega^{-1}
    )
    \sum_{i=1}^{M}
      \mV_{i1}
      \mW^{1/2}_i
      \vgamma_i
    +
    \mOmega_{\vtheta}^{-1}
    \sum_{i=1}^{M}
      \mV_{i1}
      (\mW_{\vtheta i} - \mW_i) \mW_i^{-1/2} \vgamma_i.
\]
Letting
\(
  \mE_{\vtheta i} =
    \mV_{i2}^\trans (\mbSigma_{\vtheta} - \mbSigma) \mV_{i2},
\)
the identity $(\mA + \mE)^{-1} - \mA^{-1} = -(\mA + \mE)^{-1} \mE \mA^{-1}$
implies that
\[
  \mW_{\vtheta i} - \mW_i
  =
  - \mW_{\vtheta i}
    \mE_{\vtheta i}
    \mW_i
  =
  - \mW_i
    \mE_{\vtheta i}
    \mW_i
    +
    \mW_{\vtheta i}
    \mE_{\vtheta i}
    \mW_i
    \mE_{\vtheta i}
    \mW_i.
\]
With this identity, it follows that the scaled difference between the two estimates
can be expressed as
\[
  \mOmega^{1/2} (\vhbeta_{\vtheta} - \vhbeta)
    = \vdelta_1(\vtheta) + \vdelta_2(\vtheta) + \vdelta_3(\vtheta),
\]
where
\begin{subequations}
\label{eqn:beta-delta}
\begin{align}
  \label{eqn:beta-delta1}
  \vdelta_1(\vtheta)
    &= 
      (
        \mOmega^{1/2} \mOmega_{\vtheta}^{-1} \mOmega^{1/2}
        -
        \mI_p
      )
      \mOmega^{-1/2}
      \sum_{i=1}^{M}
        \mV_{i1}
        \mW^{1/2}_i
        \vgamma_i,
\\
  \label{eqn:beta-delta2}
  \vdelta_2(\vtheta)
    &=
      -
      \mOmega^{1/2}
      \mOmega_{\vtheta}^{-1}
      \sum_{i=1}^{M}
        \mV_{i1}
        \mW_i
        \mE_{\vtheta i}
        \mW_i^{1/2} \vgamma_i,
\\        
  \label{eqn:beta-delta3}
  \vdelta_3(\vtheta)
    &=
      \mOmega^{1/2}
      \mOmega_{\vtheta}^{-1}
      \sum_{i=1}^{M}
        \mV_{i1}
        \mW_{\vtheta i}
        \mE_{\vtheta i}
        \mW_i
        \mE_{\vtheta i}
        \mW_i^{1/2} \vgamma_i.
\end{align}
\end{subequations}
Further, if Assumptions~\ref{asn:effect} and~\ref{asn:effect-est} are in
force, then $\E(\vgamma_i) = 0$ and $\cov(\vgamma_i) = \mI_{r_i}$.

Lemma~\ref{lem:beta-delta}, stated 
at the end of Section~\ref{sec:near-rel-efficiency}
and proved in Appendix~\ref{sec:finite-rel-efficiency-support} of the on-line
supplement, bounds the terms in~\eqref{eqn:beta-delta}. This lemma
implies that with probability at least
$1 - \varepsilon$, the following three inequalities simultaneously hold:
\begin{align*}
  \sup_{\vtheta \in \sB} \norm{\vdelta_1(\vtheta)}^2
    &\leq 48 \varepsilon^{-1} p^3 \tau^2, \\
  \sup_{\vtheta \in \sB} \norm{\vdelta_2(\vtheta)}^2
    &\leq 48 \varepsilon^{-1} p q^3 \tau^2 (1 + 4 p \tau)^2, \\
  \sup_{\vtheta \in \sB} \norm{\vdelta_3(\vtheta)}^2
    &\leq 768 \varepsilon^{-1} \rho \tau^4 (1 + 4 p \tau).
\end{align*}
The result of the proposition follows since
\(
  \norm{\mOmega^{1/2} (\vhbeta_{\vtheta} - \vhbeta)}^2
  \leq
  9\{
  \norm{\vdelta_1(\vtheta)}^2
  +
  \norm{\vdelta_2(\vtheta)}^2
  +
  \norm{\vdelta_3(\vtheta)}^2
  \}.
\)
\end{proof}

\begin{lemma}
\label{lem:beta-delta}
Let functions $\vdelta_1(\vtheta)$,
$\vdelta_2(\vtheta)$,
and
$\vdelta_3(\vtheta)$,
be defined as in~\eqref{eqn:beta-delta1}--\eqref{eqn:beta-delta3}.
If Assumptions~\ref{asn:effect}--\ref{asn:effect-est} are in force,
then for any $\varepsilon > 0$ and any parameter set~$\sB$,
\begin{align*}
  \Pr\{\sup_{\vtheta \in \sB} \norm{\vdelta_1(\vtheta)}^2
    &\geq
    16
    \varepsilon^{-1}
    p^3
    \tau^2
  \}
  \leq
  \varepsilon,
\\
  \Pr\{ \sup_{\vtheta \in \sB} \norm{\vdelta_2(\vtheta)}^2
    &\geq
    16 \varepsilon^{-1} p q^3 \tau^2 (1 + 4p \tau)^2
  \}
  \leq \varepsilon,
\\
  \Pr\{\sup_{\vtheta \in \sB} \norm{\vdelta_3(\vtheta)}^2
    &\geq 256 \varepsilon^{-1} \rho \tau^4 (1 + 4 p \tau)
  \}
  \leq \varepsilon,
\end{align*}
where
\(
  \tau =
  \sup_{\vtheta \in \sB} \max\{
    \norm{\mbSigma^{-1} \mbSigma_{\vtheta} - \mI_p},
    \norm{\mbSigma_{\vtheta}^{-1} \mbSigma - \mI_p}
  \}
\)
and
\(
  \rho = \sum_{i=1}^{M} r_i.
\)
\end{lemma}

\section{Asymptotic properties of two-step estimates}
\label{sec:asymptotics}


In Section~\ref{sec:finite-sample-prop},
we established finite-sample existence, concentration
bounds, and near relative efficiency for moment based estimates. Given the
finite-sample results, it is straightforward to derive asymptotic analogues of
these properties in settings where the sample size tends to infinity.

We will need an additional assumption on the bounding constants:


\begin{assumption}\label{asn:kappa}
The sequence of bounding constants $\kappa_N$ defined in
Assumption~\ref{asn:weight} satisfy $\limsup_N \kappa_N < \infty$.
\end{assumption}

\noindent
Referring to Table~\ref{tab:weight-constants}, we can see that
Assumption~\ref{asn:kappa} holds for the unweighted case whenever $\norm{\mD_{i}}$
is bounded away from zero, and for the weighted case whenever $\norm{\mD_{i}}$
is bounded away from infinity.  For the semi-weighted case,
Assumption~\ref{asn:kappa} holds whenever $\mbSigma_0$ is positive-definite.

In addition to assumptions on the bounding constants $\kappa_N$, the
asymptotic results require conditions on $\mOmega$ and $\mOmega_2$.  To state
these conditions, we define the quantities
\[
  \omega_N =
    \inf_{\vt \in \reals^p}
      \frac{\vt^\trans \mOmega \vt}{\vt^\trans \vt},
\qquad
  \omega_{N,2} =
    \inf_{\vs \in \sS_q}
      \frac{\vs^\trans \mOmega_2 \vs}{\vs^\trans \vs}.
\]
The quantity $\omega_N$ is the smallest eigenvalue of $\mOmega$; similarly,
$\omega_{N,2}$ is the smallest eigenvalue of $\mOmega_2$ restricted the space
$\sS_q$.
The asymptotic results require that $\omega_{N}$ and $\omega_{N,2}$
go to infinity at or above a specified rate.  Typically, a necessary condition
for $\omega_{N,2}$ to go to infinity is that $M \to \infty$.  For example, in
the unweighted and the semi-weighted case with $\mbSigma_0 \succ 0$, one
can show that $\omega_{N,2} = \Oh(M)$; thus, for $\omega_{N,2}$ to diverge to
infinity, it is necessary to have $M \to \infty$.


Our first result establishes that the moment-based estimators for
$\vbeta$ and $\mSigma$ are consistent.  This result follows immediately from
Corollary~\ref{cor:beta-concentration} and Proposition~\ref{prop:sigma-concentration}.

\begin{proposition}[Consistency]
  If Assumptions~\ref{asn:first}--\ref{asn:kappa} are in force, then the
  asymptotic limits of $\vhbeta_{\mW}$ and $\mhSigma_{\mW}$ are determined as
\begin{enumerate}
  \item If $\omega_N \to \infty$, then $\vhbeta_{\mW} \toP \vbeta$.
  \item If $\omega_{N,2} (N/\rho_N - 1) \to \infty$ and $\omega_{N,2} \to \infty$,
    then $\mhSigma_{\mW} \toP \mSigma$.
\end{enumerate}
\end{proposition}


Next we establish that the two-step estimate for $\vbeta$ is relatively
efficient.  To state this result, as in Section~\ref{eqn:optimal-weight}, 
let $\vhbeta$ be the moment-based estimate of $\vbeta$ with
variance-minimizing weights from~\eqref{eqn:optimal-weight}, and let
$\vhbeta_{\vtheta}$ be as defined in~\eqref{eqn:beta-omega}.
Proposition~\ref{prop:twostep-efficient} shows that the two-step estimator
$\vhbeta_{\vhtheta}$ is asymptotically as efficient as $\vhbeta$.
This result follows from Proposition~\ref{prop:sigma-concentration} and
Proposition~\ref{prop:finite-rel-efficiency};
Appendix~\ref{sec:twostep-efficient-proof} of the on-line supplement
gives a complete proof.

\begin{proposition}[Relative efficiency]\label{prop:twostep-efficient}
For each $N$, suppose that $\mW_{N,1}, \dotsc, \mW_{N,M(N)}$ are
weights with bounding constants $\kappa_N$ satisfying
Assumption~\ref{asn:kappa}.  Set $\vhtheta = (\hat \phi, \mhSigma_{\mW})$.
Suppose that Assumptions~\ref{asn:first}--\ref{asn:last} are in force and that
$\mSigma \succ 0$.  If $\rho_N \to \infty$, $(N - \rho_N) \log \rho_N \to
\infty$, and $(\omega_{N,2}^2 / \rho_N) \log \rho_N \to \infty$,
then
\(
  \mOmega^{1/2} (\vhbeta_{\vhtheta} - \vhbeta) \toP 0.
\)
\end{proposition}


The next two results show that the two-step
estimator $\vhbeta_{\vhtheta}$ is asymptotically normal.

\begin{proposition}\label{prop:optimal-normal}
Suppose that Assumptions~\ref{asn:first}--\ref{asn:last} are in force.  Let
$\vhbeta$ denote the weight-based moment estimate with variance-minimizing
weights $\mW_{i}$ as in Eq.~\eqref{eqn:optimal-weight}.  If $M \to \infty$ and
\(
  \sum_{i=1}^{M} \norm{\mOmega^{-1} \mV_{i1} \mW_i \mV_{i1}^\trans}^4 \to 0,
\)
then $\mOmega^{1/2}(\vhbeta - \vbeta)$ converges in distribution to a
mean-zero multivariate normal random vector with identity covariance matrix.
\end{proposition}

\begin{proof}
  By the Cram\'er-Wold device, it suffices to show that for any unit vector
$\vt$, the quantity $Y = \vt^\trans \mOmega^{1/2}(\vhbeta - \vbeta)$ converges in
distribution to a standard normal random variable.  For $i = 1, \dotsc, M$,
define
\[
  X_i
  =
    \vt^\trans \mOmega^{-1/2} \mV_{i1} \mW_i
    (\mV_i^\trans \vheta_i - \mV_{i1}^\trans \vbeta)
  =
    \vt^\trans \mOmega^{-1/2} \mV_{i1} \mW_i
    (\mV_i^\trans \vh_i + \mV_{i2}^\trans \vu_i),
\]
so that $Y = \sum_{i=1}^{M} X_i$.  It follows that $\E(Y) = 0$ and $\var(Y) =
1$.  If we can show that $\sum_{i=1}^{M} \E(X_i^4) \to 0$, then Lyapunov's
Theorem will ensure that $Y$ converges in distribution to a standard normal
random variable, the desired result of the proposition.

By the Cauchy-Schwarz inequality,
\[
  X_i^2
  \leq
  \norm{\mOmega^{-1/2} \mV_{i1} \mW_i (\mV_i^\trans \vh_i + \mV_{i2}^\trans \vu_i)}^2
  \leq
  \norm{\mOmega^{-1/2} \mV_{i1} \mW_i^{1/2}}^2
  \norm{\mW_i^{1/2} (\mV_i^\trans \vh_i + \mV_{i2}^\trans \vu_i)}^2.
\]
Therefore, it follows that
\[
  \E(X_i^4)
  \leq
  \norm{\mOmega^{-1/2} \mV_{i1} \mW_i^{1/2}}^4
  \, \E\norm{\mW_i^{1/2} (\mV_i^\trans \vh_i + \mV_{i2}^\trans \vu_i)}^4.
\]
One can write
\(
  \norm{\mOmega^{-1/2} \mV_{i1} \mW_i^{1/2}}^4
  =
  \norm{\mOmega^{-1} \mV_{i1} \mW_i \mV_{i1}^\trans}^2.
\)
From Assumptions~\ref{asn:effect} and~\ref{asn:effect-est}, it follows that
\(
  \E\norm{\mW_i^{1/2} (\mV_i^\trans \vh_i + \mV_{i2}^\trans \vu_i)}^4
  \leq C
\)
for some constant $C$ independent of $N$.  Thus, if
\(
  \sum_{i=1}^{M} \norm{\mOmega^{-1} \mV_{i1} \mW_i \mV_{i1}^\trans}^4 \to 0,
\)
then
\(
  \sum_{i=1}^{M} \E(X_i^4) \to 0,
\)
and hence $Y$ converges in distribution to a standard normal random variable.
\end{proof}

\begin{corollary}[Asymptotic normality]\label{cor:twostep-normal}
If the assumptions of Propositions~\ref{prop:twostep-efficient}
and~\ref{prop:optimal-normal} are in force, then the vector
$\mOmega^{1/2} (\vhbeta_{\vhtheta} - \vbeta)$ converges in distribution to a
mean-zero multivariate normal random vector with identity covariance.
\end{corollary}

\section{Performance in simulations}

\label{sec:simulations}
\label{sec:hglm-simulation}

To evaluate the performance of the moment-based estimators in practice, and to
compare these estimators to their likelihood-based counterparts, we perform
two simulation studies: one for a hierarchical linear regression model, and
one for a hierarchical logistic regression model.  This section describes the
logistic regression simulation; Appendix~\ref{sec:hlm-simulation} of the
on-line supplement describes the linear regression case.  Both simulations exhibit similar behaviors.

\begin{figure}
  \centering
  \makebox[\textwidth][c]{
    \includegraphics[width=15cm]{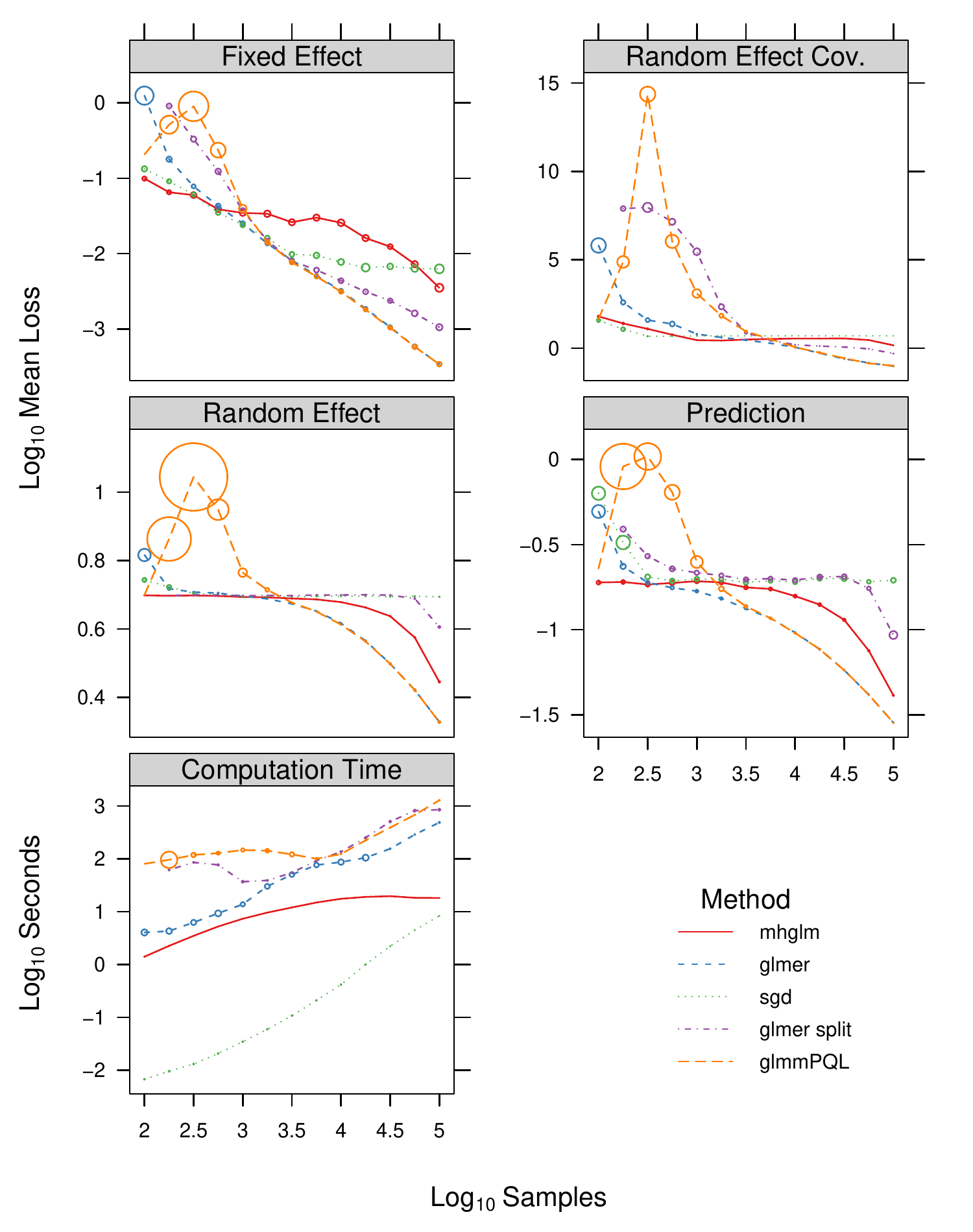}
  }
  \caption{
    Performance for the hierarchical logistic model.  Circle radii indicate
    one standard error along $y$-axis (absent when smaller than line width).
  }
  \label{fig:sim-logistic}
\end{figure}

We set the number of groups to $M = 1000$ and simulate $N$
samples, with $N$ ranging from $100$ to $100000$.  We set the dimensions of
the fixed and random effect vectors to $p = q = 5$.  For each value of $N$ we
draw $100$ replicates according to the following procedure.

For each replicate, we draw a $p$-dimensional fixed effect vector $\vbeta$
with components $\beta_k$, $k = 1, \dotsc, p$ drawn independently from a $t$
distribution with $4$ degrees of freedom.  We draw random effect covariance
matrix $\mSigma$ from an inverse Wishart distribution with shape $\mI$ and $2
q$ degrees of freedom, scaled by $0.1$.

Rather than splitting the $N$ samples evenly across all $M$ groups, in each
replicate we draw population-specific sampling rates $\lambda_i\ (i = 1,
\dotsc, M)$ as independent exponential random variables with mean $N/M$.
Then, we allocate the $N$ sample points by drawing from a multinomial on $M$
categories with probability of category~$i$ proportional to $\lambda_i$.  This
sampling scheme is equivalent to drawing $n_1, \dotsc, n_M$ as independent
geometric random variables with mean $N/M$, conditional on their sum being
$N$; it gives rise to a highly skewed distribution of sample sizes.

For each group $i = 1, \dotsc, M$, once $n_i$ has been determined we draw a
random effect vector~$\vu_i$ as multivariate normal random vector with mean
zero and covariance $\mSigma$.  We draw random population-specific fixed
effect predictor vectors $\vx_{ij}$ for $j = 1, \dotsc, n_i$ with independent
elements such that $\Pr(x_{ijk} = + 1) = \Pr(x_{ijk} = -1) = 1/2$ for $k = 1,
\dotsc, p$.  We use the same procedure to random effect predictor vectors
$\vz_{ij}$. Finally, for $j = 1, \dotsc, n_i$, we draw response variate
$y_{ij}$ as Bernoulli with success probability $\mu_{ij} = \logit^{-1}
(\vx_{ij}^\trans \vbeta + \vz_{ij}^\trans \vu_i)$.

We use a variety of methods to compute estimates of the population parameters
$\vbeta$ and $\mSigma$, along with plug-in empirical Bayes estimates
group-specific random effects $\vhu_i$, $i = 1, \dotsc, M$:
\begin{enumerate}

  \item \emph{mhglm}, the proposed moment-based estimation procedure.  To
    compute the moment-based estimates, we use two-step estimators with
    semi-weighted initial step and $\mSigma_0$ set to the identity
    matrix, after standardizing the predictors.
    The procedure is implemented in the R programming language.

  \item \emph{glmer}, maximum likelihood, using a gradient-free optimization
    procedure applied to an order-0 Laplace approximation to the profiled likelihood,
    implemented in C++ and R by the \emph{lme4} R package \citep{lme4}.

  \item \emph{sgd}, which uses stochastic gradient descent to maximize a
    regularized version of the $h$-likelihood (described in detail in
    Appendix~\ref{sec:sgd} of the on-line supplement). The compute-intensive
    inner loop is implemented in C, and the outer loop in~R.
    
  \item \emph{glmer split}, a data-splitting estimation procedure, which splits
    the data set into 10 subsets, computes separate estimates for each using
    \emph{glmer}, and then combines the estimates by averaging them.
    Implemented in R.

  \item \emph{glmmPQL}, penalized quasi-likelihood, as implemented by the
    \emph{MASS} package by iteratively calling the \emph{lme} fitting
    procedure \citep{MASS}.

\end{enumerate}
We report serial computation time for each procedure, and we do not include
cross-validation time for the tuning parameter selection for the \emph{sgd}
method.

To evaluate the performances of the estimators, we use
\(
  \norm{\vbeta - \vhbeta}^2
\)
for the fixed effect loss,
\(
  \tr\{(\mhSigma \mSigma^{-1} - \mI)^2\}
\)
for the random effet covariance loss,
\(
  M^{-1} \sum_{i=1}^{M} \norm{\mSigma^{-1/2} (\vu_i - \vhu_i)}^2
\)
for the random effect loss, and
\(
  2 N^{-1} \sum_{i=1}^{M} \sum_{j=1}^{n_i} [
    \mu_{ij} \log(\mu_{ij} / \hat \mu_{ij})
    +
    (1 - \mu_{ij}) \log\{(1 - \mu_{ij}) / (1 - \hat \mu_{ij})\}
  ]
\)
for the prediction loss,
where $\mu_{ij} = \logit^{-1} (\vx_{ij}^\trans \vbeta + \vz_{ij}^\trans \vu_i)$
and $\hat \mu_{ij} = \logit^{-1} (\vx_{ij}^\trans \vhbeta + \vz_{ij}^\trans \vhu_i)$.


Fig.~\ref{fig:sim-logistic} shows the mean loss, averaged over all replicates,
with circle radii indicating standard errors along the vertical axes (absent
when less than the visible line width).  For moderate to large sample sizes,
there is a noticeable loss in statistical efficiency between the proposed
method (\emph{mhglm}) and the methods based on maximum likelihood
(\emph{glmer} and \emph{glmmPQL}).  Still, the proposed method appears to be
consistent.  Moreover, in terms of prediction loss, it performs better than
\emph{glmer split} and \emph{sgd}.

The lower-left panel of Fig.~\ref{fig:sim-logistic} shows the sequential
computation times for all methods.  For the largest values of $N$ tried in the
simulation, the proposed method is faster than the exact and approximate
maximum likelihood procedures by factor ranging from $10$ to $100$.  Without
including cross-validation time, the \emph{sgd} method is faster than all
other methods tried in the simulation.

In this simulation, it appears that the \emph{sgd} method trades substantial
statistical efficiency for improvements in computational efficiency.  The
proposed \emph{mhglm} method makes a similar trade-off, but delivers
noticeably higher statistical efficiency.

\section{Application to recommender systems}
\label{sec:application}

\subsection{Motivation}

To demonstrate the potential utility of the proposed moment-based estimators,
we apply them to a large-scale recommender system application.  Specifically,
we use them to fit a hierarchical model to the MovieLens~10M dataset: the
$N = 10000054$ ratings of $M = 69878$ users for $10681$ movies
\citep{Movi09}.  Using a moment-based
estimation procedure to fit a hierarchical model to this dataset required
approximately 10 minutes of serial computation time; the $glmer$ method
required approximately 9 hours to fit the same model.
In Sections~\ref{sec:estimating-user-prefs}--\ref{sec:predicting-user-ratings}
we demonstrate the ability of a hierarchical model, fitted using moment based
estimation, to estimate user preferences and predict user ratings.

\subsection{Estimating user preferences}

\label{sec:estimating-user-prefs}

One goal with a recommender system is to estimate user-specific preferences.
This information can be used to characterize the user population and to
cluster the users into meaningful groups, possibly for targeting promotions or
advertisements.  Formally, we represent a user's preferences by a vector of
coefficients which relate observable covariates to the user's ratings.
We will try to estimate these user-specific coefficients from the available
movie rating data.

Each rating consists of a user, and movie, a time, and a star value between 0
and 5.  We binarize the ratings, 
then use a logistic regression model to relate the binarized
ratings to review-specific predictors.  We use the same predictors for the
fixed and random effects, so that the model reduces to a random coefficient
model. Letting $\vbeta_i = \vbeta + \vu_i$ be a user-specific coefficient vector
(fixed plus random effect), the model specifies
\(
\logit \Pr( y_{ij} = 1 \mid \vx_{ij}, \vu_i)
  =
  \vx_{ij}^\trans \vbeta + \vx_{ij}^\trans \vu_i
  =
  \vx_{ij}^\trans \vbeta_i
\)
where $y_{ij}$ indicates whether or not rating $ij$ is favourable (at least 4
stars) and $\vx_{ij}$ is a set of rating predictors.

\begin{table}
\caption{Predictors associated with review $ij$}
\label{tab:ml-10m-features}

\centering
\begin{tabular}{l p{.70\textwidth}}
\toprule
Predictor & Description \\
\midrule
$\text{Genre}_{ij}$
& A 4-component vector with movie-specific genre scores for Action,
    Children, Comedy, and Drama of the rated movie. Movies belonging to
    multiple genres have fractional scores for individual categories.
We use effect coding, so that the coefficients for the $4$ genre
components sum to zero.
\\
$\text{Popularity}_{ij}$
& A robust estimate of the logit of the current popularity of the rated movie,
computed from recent ratings of the movie:
$\logit \{ (l_{ij} + 0.5) / (n_{ij} + 1.0)\}$, where $l_{ij}$ is the number of
users who recently liked the movie and $n_{ij}$ is the number of recent
reviews of the movie.  Here, ``recent'' reviews of the movie are the 30 or
fewer most recent reviews at the time of rating $ij$.
\\
$\text{Previous}_{ij}$
& An indicator of whether or not user $i$ gave a favourable star
value ($\geq 4$) in his or her previous rating.
This predictor is designed to capture the user's current overall mood.  
\\
\bottomrule
\end{tabular}

\end{table}

\begin{figure}
  \centering
  \makebox[\textwidth][c]{
    \includegraphics[width=13.5cm]{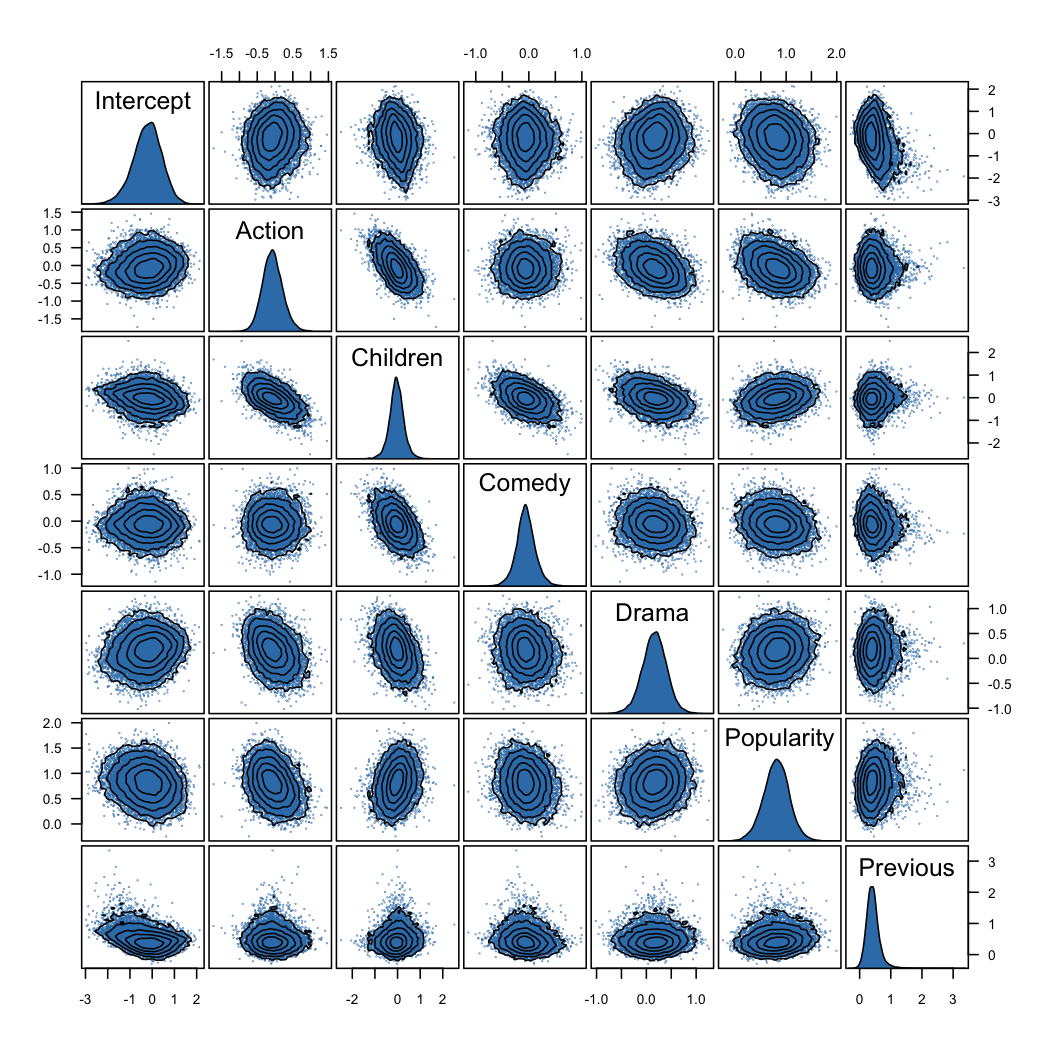}
  }
  \caption{Empirical Bayes coefficient estimates for the 26884 users with at
  least 100 reviews.}
  \label{fig:ml-10m-coef-pairs}
\end{figure}

%
%

Our first set of predictors encodes the genre of the movie being rated.
The remaining rating-specific predictors are motivated by
intuition derived from the BellKor movie recommender system \citep{Kore09}.
One predictor, $\text{Popularity}_{ij}$ captures the current popularity of the
movie being rating.
The other predictor, $\text{Previous}_{ij}$, indicates whether or not
the user's previous rating was positive; 
Table~\ref{tab:ml-10m-features} describes
these predictors in detail.


We assume a hierarchical model for the coefficient vectors with $\E(\vbeta_i) =
\vbeta$ and $\cov(\vbeta_i) = \mSigma$ for $i = 1,\dotsc, M$.  We use
moment-based estimators for $\vbeta$ and $\mSigma$ computed from all $N$
ratings, and then compute approximate empirical Bayes
estimates for $\vbeta_i\ (i = 1, \dotsc, M)$ assuming that the coefficients
come from a multivariate normal population.
Fig.~\ref{fig:ml-10m-coef-pairs} shows the one- and two-dimensional marginal
distributions of the empirical Bayes coefficient estimates for those users
with at least 100 ratings.  In the two-dimensional marginals, contour lines
show approximately 38\%, 68\%, 87\%, 95\%, and 99\% of coefficient pairs;
these lines should be elliptical and evenly spaced for bivariate
normally-distributed pairs.  For the most part, the bivariate distributions
look approximately normal, excepting the coefficient of
$\text{Previous}_{ij}$.  

By looking at the associations between the estimated coefficients, we can
conclude that (a)~affinity for particular genres appears
unrelated to the intercept, which encodes a user's overall tendency to give
positive ratings;
(b)~users who like action movies tend to
dislike children's and drama movies, users who like children's movies tend to
dislike other genres, and users who like drama movies tend to dislike action and
children's movies;
(c)~users who like action movies tend to prefer unpopular movies, and users who
like children's movies tend to prefer popular movies;
(d)~users who tend to give ratings similar to their previous ratings do not
tend to have preferences for particular genres.
Not only does the hierarchical coefficient model allow for a
diversity of user preferences (encoded in regression coefficients), it also
reveals associations between these preferences.

\subsection{Predicting user ratings}

\label{sec:predicting-user-ratings}

\begin{figure}
  \centering
  \includegraphics[height=7.5cm]{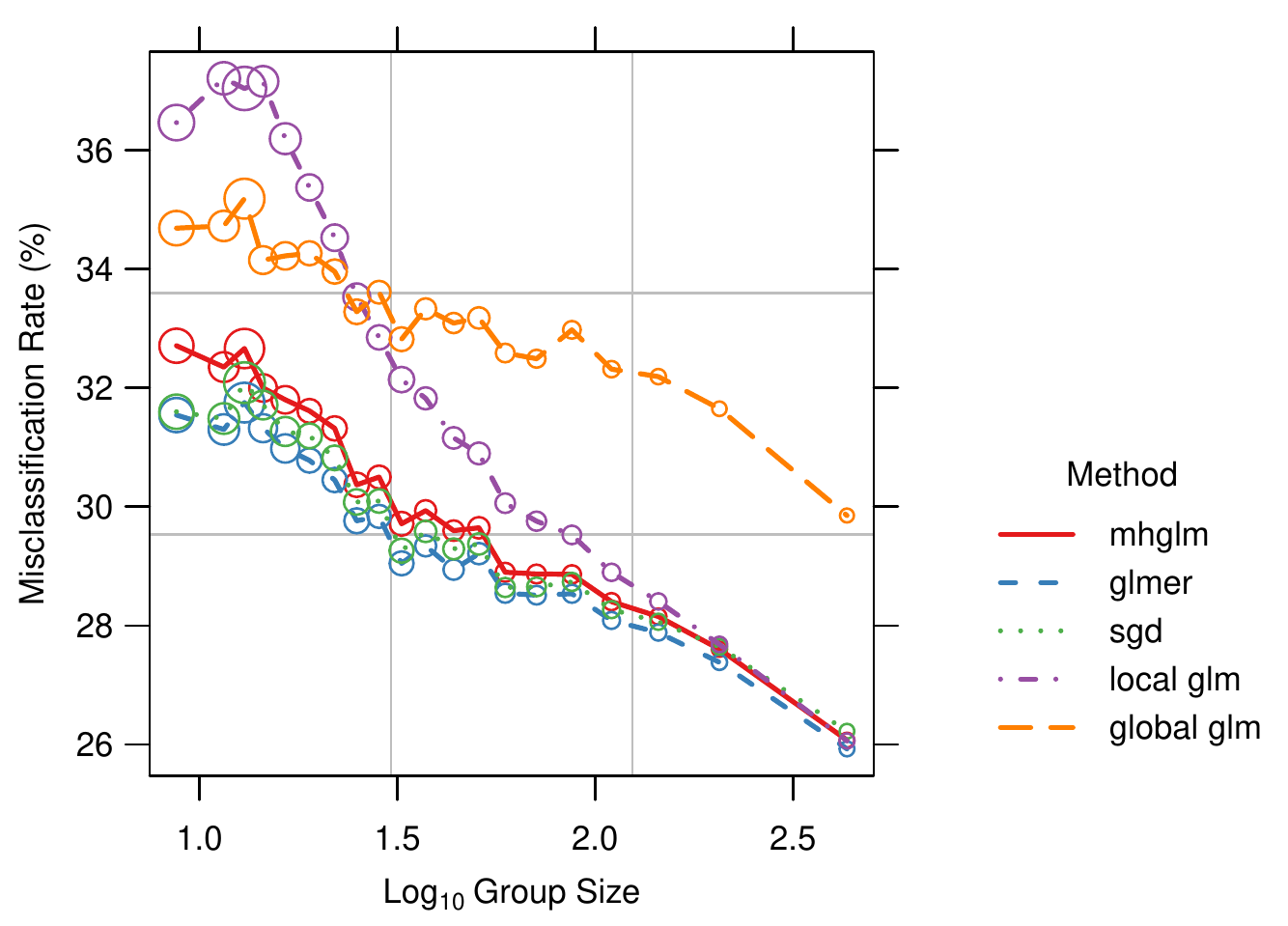}
  \caption{Misclassification rate for each user, $i$, aggregated by group size, $n_i$.}
  \label{fig:ml-10m-misclass}
\end{figure}

Often, the primary goal of a recommender system is to predict item ratings.
For this task, one advantage a hierarchical method holds over competing
methods is its ability to borrow estimation strength across similar users,
often obtaining better estimates than a model which estimates user-specific
coefficients in isolation.  To demonstrate this ability, we compare the
out-of-sample prediction performances of three models: a ``global''
generalized linear model, using a single coefficient vector for all users,
estimated by Firth's penalized maximum likelihood; a ``local'' generalized
linear model, which uses separate
coefficient vectors for all users, independently estimated with user-specific
data and penalized maximum likelihood; and a hierarchical logistic regression
model, which uses approximate empirical Bayes posterior means of the coefficients
in the hierarchical model.  We fit the hierarchical model using three
different methods: moment-based estimation (\emph{mhglm}), maximum profile
likelihood (\emph{glmer}), and stochastic gradient descent (\emph{sgd}).

We randomly split the reviews into 50\% for a training set and 50\% for a test
set.  We fit all three models on the training set, then use the fitted models
to predict the values in the test set.
Fig.~\ref{fig:ml-10m-misclass} shows
the misclassification loss performances of the fitted models on the test set for each user $i$,
aggregated by group size, $n_i$.  The lines shows the averages, and the radii
of the circles indicate standard errors along with $y$-axis.  All three
fitting methods for the hierarchical models perform comparably.
The hierarchical methods uniformly beat the local and the global models.
By combining the flexibility of the local model with the
stability of the global model, the hierarchical model is able to outperform
both extremes.

\section{Discussion}
\label{sec:discussion}

We have extended Cochran's moment-based estimators to general hierarchical
models.  Unlike other extensions, our proposal allows for both fixed and random
effects, and it accommodates rank-degenerate predictor matrices.  The proposed
estimation procedure has three main properties which make it appealing in
large-scale data regimes.  First, the procedure does not rely on strong
distributional assumptions.  Second, even when distributional assumptions are
in force, in large sample settings the method can exhibit estimation and
prediction performance comparable to likelihood-based estimators.  Finally,
and most importantly, the method has good computational performance, sometimes
10 to 100 times faster than existing maximum likelihood procedures.

We have analyzed the proposed method, both theoretically and empirically.  We
have shown that, subject to mild regularity assumptions, the moment-based
estimation procedure is consistent.  Moreover, the two-step estimation
procedure is asymptotically relatively efficient and asymptotically normal,
facilitating inference for the fixed effect vector.

The assumptions required for the theoretical results hold for most
hierarchical linear models.  However, for hierarchical generalized linear
models, these assumptions will only be good approximations when the
group-specific sample sizes $n_i$ are large; when this is not the case, the
theoretical consistency results will no longer apply.  In
Sections~\ref{sec:hglm-simulation} and~\ref{sec:application}, we demonstrate that
even without theoretical guarantees, the proposed method can perform well.  It
is an open question to derive exact theoretical conditions to guarantee that
the moment-based estimators for hierarchical generalized linear models are
consistent.


It is natural to ask if the moment-based estimators discussed in this article
can be extended to handle more general models.  For more general hierarchical
models with additional levels of hierarchy, this extension seems feasible, but
implementing this procedure in practice and deriving the appropriate
theoretical conditions to guarantee consistency will require some finesse.

To extend the proposed estimators to more general mixed models with non-nested
random effects, it is not obvious how to proceed.  We rely crucially on the
ability to get conditionally independent subpopulation-specific coefficient
estimates.  This is likely impossible with crossed random effects.  In our
recommender system application, we were able to obviate the need for
item-specific random effects by introducing a data-dependent predictor to
capture item popularity.  While this is not a perfect solution, it falls
within our modelling framework, and it is simple to implement.   It is likely
that similar predictors can be used in other contexts
where one would normally use crossed random effects.

As data volumes continue to outpace computational capacity, it becomes
increasingly advantageous to trade statistical for computational efficiency.
This is sometimes difficult, and it is only achievable if computational
demands are a primary concern throughout the development of the methodology.
We have demonstrated that when using moment-based estimates for hierarchical
models, it is sometimes possible to gain substantial improvements in speed
without sacrificing too much estimation performance.

\section*{Acknowledgement}

The author thanks Brendan O'Connor, Marc Scott, Jeff Simonoff, and the
anonymous referees for providing references and for
suggesting edits that greatly improved the article.

\appendix

\bibliographystyle{chicago}
\bibliography{refs}




\section*{Supplementary material}

# Fast moment-based estimation for hierarchical models: Supplementary appendices

Patrick O. Perry

*Stern School of Business, New York University, USA*

## A. Proof of Proposition 5.3

Note that the matrix $\boldsymbol{W}_i \otimes \boldsymbol{W}_i$ is symmetric and positive-semidefinite. Using the decomposition

$$\boldsymbol{\Omega}_2 = \sum_{i=1}^{M} (\boldsymbol{V}_{i2} \otimes \boldsymbol{V}_{i2})(\boldsymbol{W}_i \otimes \boldsymbol{W}_i)(\boldsymbol{V}_{i2} \otimes \boldsymbol{V}_{i2})^{\mathrm{T}},$$

it follows that $\boldsymbol{\Omega}_2$ is symmetric and positive semidefinite. Further, one can show that $\mathcal{S}_q$ is an invariant subspace of $\boldsymbol{\Omega}_2$. Applying Theorem 8.1.9 from Golub and Van Loan (1996), it suffices to show that $\boldsymbol{s}^{\mathrm{T}} \boldsymbol{\Omega}_2 \boldsymbol{s} > 0$ for all $\boldsymbol{s} \in \mathcal{S}_q$. Suppose that the converse is true, so that there exists a nonzero vector $\boldsymbol{s}$ with $\boldsymbol{s}^{\mathrm{T}} \boldsymbol{\Omega}_2 \boldsymbol{s} = 0$. Let $\boldsymbol{S}$ be the $q \times q$ matrix with $\mathrm{vec}(\boldsymbol{S}) = \boldsymbol{s}$, so that

$$0 = \boldsymbol{s}^{\mathrm{T}} \boldsymbol{\Omega}_2 \boldsymbol{s} = \sum_{i=1}^{M} \| \boldsymbol{W}_i^{1/2} \boldsymbol{V}_{i2} \boldsymbol{S} \boldsymbol{V}_{i2}^{\mathrm{T}} \boldsymbol{W}_i^{1/2} \|_{\mathrm{F}}^2,$$

where $\|\cdot\|_{\mathrm{F}}$ denotes Frobenius norm. It follows that $\boldsymbol{V}_{i2}^{\mathrm{T}} \boldsymbol{S} \boldsymbol{V}_{i2} = 0$ for all $i = 1, \ldots, M$. In particular,

$$0 = \sum_{i=1}^{M} \| \boldsymbol{V}_{i2} \boldsymbol{S} \boldsymbol{V}_{i2}^{\mathrm{T}} \|_{\mathrm{F}}^2 = \boldsymbol{s}^{\mathrm{T}} \Big\{ \sum_{i=1}^{M} (\boldsymbol{V}_{i2} \boldsymbol{V}_{i2}^{\mathrm{T}}) \otimes (\boldsymbol{V}_{i2} \boldsymbol{V}_{i2}^{\mathrm{T}}) \Big\} \boldsymbol{s}$$

This contradicts assumption 2. Therefore, we must have $\boldsymbol{s}^{\mathrm{T}} \boldsymbol{\Omega}_2 \boldsymbol{s} > 0$ for all $\boldsymbol{s} \in \mathcal{S}_q$, so that $\boldsymbol{\Omega}_2$ is invertible on $\mathcal{S}_q$.

## B. Proof of Lemma 5.7

We prove Lemma 5.7 parts (a)–(d) through a series of smaller lemmas, labeled as Lemmas B.2–B.5. These smaller results rely on a matrix version of the Cauchy-Schwarz inequality, which we state and prove as Lemma B.1.

LEMMA B.1 (MATRIX CAUCHY-SCHWARZ INEQUALITY). *Suppose matrices $\boldsymbol{A}_1, \ldots, \boldsymbol{A}_M$ and $\boldsymbol{B}_1, \ldots, \boldsymbol{B}_M$ are such that $\boldsymbol{A}_i$ has dimension $p \times r_i$ and $\boldsymbol{B}_i$ has dimension $q \times r_i$. Define block matrices $\boldsymbol{A} = [\boldsymbol{A}_1^{\mathrm{T}} \cdots \boldsymbol{A}_M^{\mathrm{T}}]^{\mathrm{T}}$ and $\boldsymbol{B} = [\boldsymbol{B}_1^{\mathrm{T}} \cdots \boldsymbol{B}_M^{\mathrm{T}}]^{\mathrm{T}}$. Then,*

$$\| \boldsymbol{A}^{\mathrm{T}} \boldsymbol{B} \| \leq \| \boldsymbol{A}^{\mathrm{T}} \boldsymbol{A} \|^{1/2} \| \boldsymbol{B}^{\mathrm{T}} \boldsymbol{B} \|^{1/2}.$$

*Address for correspondence:* Patrick O. Perry, Information, Operations, and Management Sciences Department, Stern School of Business, New York University, 44 West 4th St, New York, NY 10012, USA
E-mail: pperry@stern.nyu.edu

PROOF. We first will prove a special case of the result assuming that $\boldsymbol{A}^\mathrm{T}\boldsymbol{A} = \boldsymbol{I}_p$ and $\boldsymbol{B}^\mathrm{T}\boldsymbol{B} = \boldsymbol{I}_q$. Let $\boldsymbol{A}^\mathrm{T}\boldsymbol{B} = \boldsymbol{U}\boldsymbol{D}\boldsymbol{V}^\mathrm{T}$ be a compact singular value decomposition, where the left and right singular vector matrices satisfy $\boldsymbol{U}^\mathrm{T}\boldsymbol{U} = \boldsymbol{V}^\mathrm{T}\boldsymbol{V} = \boldsymbol{I}_r$, and $\boldsymbol{D}$ is an $r \times r$ diagonal matrix with positive diagonal entries. Note that $\|\boldsymbol{A}^\mathrm{T}\boldsymbol{B}\| = \|\boldsymbol{D}\|$. Define $\tilde{\boldsymbol{A}} = \boldsymbol{A}\boldsymbol{U}$ and $\tilde{\boldsymbol{B}} = \boldsymbol{B}\boldsymbol{V}$. Then, it follows that $\tilde{\boldsymbol{A}}^\mathrm{T}\tilde{\boldsymbol{A}} = \tilde{\boldsymbol{B}}^\mathrm{T}\tilde{\boldsymbol{B}} = \boldsymbol{I}_r$. Hence,

$$0 \preceq (\tilde{\boldsymbol{A}} - \tilde{\boldsymbol{B}})^\mathrm{T}(\tilde{\boldsymbol{A}} - \tilde{\boldsymbol{B}}) = 2(\boldsymbol{I} - \boldsymbol{D}).$$

Thus, $\|\boldsymbol{D}\| \leq 1$.

In the general case, let $\boldsymbol{A} = \boldsymbol{U}_\boldsymbol{A}\boldsymbol{D}_\boldsymbol{A}\boldsymbol{V}_\boldsymbol{A}^\mathrm{T}$ and $\boldsymbol{B} = \boldsymbol{U}_\boldsymbol{B}\boldsymbol{D}_\boldsymbol{B}\boldsymbol{V}_\boldsymbol{B}^\mathrm{T}$ be compact singular value decompositions. Then,

$$\boldsymbol{A}^\mathrm{T}\boldsymbol{B} = \boldsymbol{V}_\boldsymbol{A}\boldsymbol{D}_\boldsymbol{A}\boldsymbol{U}_\boldsymbol{A}^\mathrm{T}\boldsymbol{U}_\boldsymbol{B}\boldsymbol{D}_\boldsymbol{B}\boldsymbol{V}_\boldsymbol{B}^\mathrm{T},$$

and so

$$\|\boldsymbol{A}^\mathrm{T}\boldsymbol{B}\| \leq \|\boldsymbol{D}_\boldsymbol{A}\|\|\boldsymbol{D}_\boldsymbol{B}\|\|\boldsymbol{U}_\boldsymbol{A}^\mathrm{T}\boldsymbol{U}_\boldsymbol{B}\|.$$

The special case of the result shows that $\|\boldsymbol{U}_\boldsymbol{A}^\mathrm{T}\boldsymbol{U}_\boldsymbol{B}\| \leq 1$. The general result follows from the identities $\|\boldsymbol{D}_\boldsymbol{A}\| = \|\boldsymbol{A}^\mathrm{T}\boldsymbol{A}\|^{1/2}$ and $\|\boldsymbol{D}_\boldsymbol{B}\| = \|\boldsymbol{B}^\mathrm{T}\boldsymbol{B}\|^{1/2}$.

LEMMA B.2. *If Assumptions 1, 2, 3, and 5 are in force, then with probability at least $1 - \varepsilon$,*

$$\|\hat{\boldsymbol{S}}(\hat{\boldsymbol{\beta}}_\boldsymbol{W}) - \hat{\boldsymbol{S}}(\boldsymbol{\beta})\|_\mathrm{F} < 9p^{3/2}\kappa\{\mathrm{tr}(\boldsymbol{\Omega}_2^{-1})\}^{1/2}/\varepsilon.$$

PROOF. Define $\boldsymbol{g} = \boldsymbol{\beta} - \hat{\boldsymbol{\beta}}$ and note

$$(\boldsymbol{V}_i\hat{\boldsymbol{\eta}}_i - \boldsymbol{V}_{i1}\hat{\boldsymbol{\beta}}_\boldsymbol{W})(\boldsymbol{V}_i\hat{\boldsymbol{\eta}}_i - \boldsymbol{V}_{i1}\hat{\boldsymbol{\beta}}_\boldsymbol{W})^\mathrm{T} = (\boldsymbol{V}_i\hat{\boldsymbol{\eta}}_i - \boldsymbol{V}_{i1}\boldsymbol{\beta})(\boldsymbol{V}_i\hat{\boldsymbol{\eta}}_i - \boldsymbol{V}_{i1}\boldsymbol{\beta})^\mathrm{T} + \boldsymbol{V}_{i1}^\mathrm{T}\boldsymbol{g}\boldsymbol{g}^\mathrm{T}\boldsymbol{V}_{i1} \\ + \boldsymbol{V}_{i1}^\mathrm{T}\boldsymbol{g}(\boldsymbol{V}_i\hat{\boldsymbol{\eta}}_i - \boldsymbol{V}_{i1}\boldsymbol{\beta})^\mathrm{T} + (\boldsymbol{V}_i\hat{\boldsymbol{\eta}}_i - \boldsymbol{V}_{i1}\boldsymbol{\beta})\boldsymbol{g}^\mathrm{T}\boldsymbol{V}_{i1}.$$

Thus, the difference of the weighted sums can be written as

$$\mathrm{vec}\{\hat{\boldsymbol{A}}(\hat{\boldsymbol{\beta}}_\boldsymbol{W})\} - \mathrm{vec}\{\hat{\boldsymbol{A}}(\boldsymbol{\beta})\} = \boldsymbol{\Delta}_1 + \boldsymbol{\Delta}_2 + \boldsymbol{\Delta}_2^\mathrm{T},$$

where

$$\boldsymbol{\Delta}_1 = \sum_{i=1}^M \boldsymbol{V}_{i2}\boldsymbol{W}_i\boldsymbol{V}_{i1}^\mathrm{T}\boldsymbol{g}\boldsymbol{g}^\mathrm{T}\boldsymbol{V}_{i1}\boldsymbol{W}_i\boldsymbol{V}_{i2}^\mathrm{T},$$

$$\boldsymbol{\Delta}_2 = \sum_{i=1}^M \boldsymbol{V}_{i2}\boldsymbol{W}_i\boldsymbol{V}_{i1}^\mathrm{T}\boldsymbol{g}(\boldsymbol{V}_i\hat{\boldsymbol{\eta}}_i - \boldsymbol{V}_{i1}\boldsymbol{\beta})^\mathrm{T}\boldsymbol{W}_i\boldsymbol{V}_{i2}^\mathrm{T}.$$

Defining $\tilde{\boldsymbol{g}} = \boldsymbol{\Omega}^{1/2}\boldsymbol{g}$, it follows that

$$\mathrm{vec}\,(\boldsymbol{\Delta}_1) = \Big\{\sum_{i=1}^M \boldsymbol{V}_{i2}\boldsymbol{W}_i\boldsymbol{V}_{i1}^\mathrm{T} \otimes \boldsymbol{V}_{i2}\boldsymbol{W}_i\boldsymbol{V}_{i1}^\mathrm{T}\Big\}(\boldsymbol{\Omega}\otimes\boldsymbol{\Omega})^{-1/2}\,\mathrm{vec}(\tilde{\boldsymbol{g}}\tilde{\boldsymbol{g}}^\mathrm{T}).$$

One can show that

$$(\boldsymbol{\Omega}\otimes\boldsymbol{\Omega})^{-1/2} \preceq \Big\{\sum_{i=1}^M \boldsymbol{V}_{i1}\boldsymbol{W}_i\boldsymbol{V}_{i1}^\mathrm{T} \otimes \boldsymbol{V}_{i1}\boldsymbol{W}_i\boldsymbol{V}_{i1}^\mathrm{T}\Big\}^{-1/2}.$$

Applying the matrix Cauchy-Schwarz inequality (Lemma B.1), we get that

$$\|\boldsymbol{\Omega}_2^{-1/2}\,\mathrm{vec}\,(\boldsymbol{\Delta}_1)\| \leq \|\mathrm{vec}\,(\tilde{\boldsymbol{g}}\tilde{\boldsymbol{g}}^\mathrm{T})\| = \|\tilde{\boldsymbol{g}}\|^2.$$

From Proposition 5.4, we have that $E\|\tilde{\boldsymbol{g}}\|^2 \leq p\kappa$. Markov's inequality implies that for any $a > 0$,

$$\Pr(\|\tilde{\boldsymbol{g}}\|^2 \geq ap\kappa) \leq 1/a.$$

Thus, there exists an event $G_a$ with probability at least $1 - 1/a$ on which $\|\tilde{\boldsymbol{g}}\|^2 < ap\kappa$; on this event, the bound

$$\|\boldsymbol{\Omega}_2^{-1}\operatorname{vec}(\boldsymbol{\Delta}_1)\| < ap\kappa\|\boldsymbol{\Omega}_2^{-1}\|^{1/2} < ap\kappa\{\operatorname{tr}(\boldsymbol{\Omega}_2^{-1})\}^{1/2}$$

holds.

The term $\boldsymbol{\Delta}_2$ can be written in vector form as

$$\operatorname{vec}(\boldsymbol{\Delta}_2) = \sum_{k=1}^{p} \tilde{g}_k \boldsymbol{\delta}_{2k},$$

where $\tilde{g}_k = \boldsymbol{e}_k^{\mathrm{T}}\tilde{\boldsymbol{g}}$ with $\boldsymbol{e}_k$ the $k$th standard basis vector, and

$$\boldsymbol{\delta}_{2k} = \sum_{i=1}^{M} (\boldsymbol{V}_{i2}\boldsymbol{W}_i) \otimes (\boldsymbol{V}_{i2}\boldsymbol{W}_i\boldsymbol{V}_{i1}\boldsymbol{\Omega}^{-1/2})\operatorname{vec}\{\boldsymbol{e}_k(\boldsymbol{V}_i\hat{\boldsymbol{\eta}}_i - \boldsymbol{V}_{i1}\boldsymbol{\beta})^{\mathrm{T}}\}.$$

Further, Assumption 3 implies that $E(\boldsymbol{\delta}_{2k}) = 0$ and

$$\operatorname{cov}(\boldsymbol{\delta}_{2k}) = \sum_{i=1}^{M}\{\boldsymbol{V}_{i2}\boldsymbol{W}_i(\boldsymbol{V}_{i2}^{\mathrm{T}}\boldsymbol{\Sigma}\boldsymbol{V}_{i2} + \phi\boldsymbol{D}_i^{-2})\boldsymbol{W}_i\boldsymbol{V}_{i2}^{\mathrm{T}}\} \otimes (\boldsymbol{V}_{i2}\boldsymbol{W}_i\boldsymbol{V}_{i1}\boldsymbol{\Omega}^{-1/2}\boldsymbol{e}_k\boldsymbol{e}_k^{\mathrm{T}}\boldsymbol{\Omega}^{-1/2}\boldsymbol{V}_{i1}^{\mathrm{T}}\boldsymbol{W}_i\boldsymbol{V}_{i2}^{\mathrm{T}}).$$

Since

$$\boldsymbol{W}_i^{1/2}\boldsymbol{V}_{i1}\boldsymbol{\Omega}^{-1/2}\boldsymbol{e}_k\boldsymbol{e}_k^{\mathrm{T}}\boldsymbol{\Omega}^{-1/2}\boldsymbol{V}_{i1}^{\mathrm{T}}\boldsymbol{W}_i^{1/2} \preceq \boldsymbol{W}_i^{1/2}\boldsymbol{V}_{i1}\boldsymbol{\Omega}^{-1}\boldsymbol{V}_{i1}^{\mathrm{T}}\boldsymbol{W}_i^{1/2} \preceq \boldsymbol{I}_{r_i},$$

Assumption 5 implies that $\operatorname{cov}(\boldsymbol{\delta}_{2k}) \preceq \kappa\boldsymbol{\Omega}_2$, and so

$$E\|\boldsymbol{\Omega}_2^{-1}\boldsymbol{\delta}_{2k}\|^2 \leq \kappa\operatorname{tr}(\boldsymbol{\Omega}_2^{-1}).$$

Markov's inequality implies that

$$\Pr\{\|\boldsymbol{\Omega}_2^{-1}\boldsymbol{\delta}_{2k}\|^2 \geq ap\kappa\operatorname{tr}(\boldsymbol{\Omega}_2^{-1})\} \leq 1/(ap).$$

Since this holds for all $k = 1, \ldots, p$, it follows that there exists an event $B_a$ with probability at least $1 - 1/a$ on which

$$\sum_{i=1}^{p}\|\boldsymbol{\Omega}_2^{-1}\boldsymbol{\delta}_{2k}\|^2 < ap^2\kappa\operatorname{tr}(\boldsymbol{\Omega}_2^{-1}).$$

The matrix Cauchy-Schwartz inequality (Lemma B.1) implies that on $G_a \cap B_a$,

$$\|\boldsymbol{\Omega}_2^{-1}\operatorname{vec}(\boldsymbol{\Delta}_2)\| < (ap\kappa)^{1/2}[ap^2\kappa\operatorname{tr}(\boldsymbol{\Omega}_2^{-1})]^{1/2} = ap^{3/2}\kappa\{\operatorname{tr}(\boldsymbol{\Omega}_2^{-1})\}^{1/2}.$$

A similar argument shows that there exists an event $B_a'$ with probability at least $1 - 1/a$ such that on $G_a \cap B_a'$, the norm $\|\boldsymbol{\Omega}_2^{-1}\operatorname{vec}(\boldsymbol{\Delta}_2^{\mathrm{T}})\|$ is bounded by the same quantity.

Finally, the triangle inequality impliest that on the event $G_a \cap B_a \cap B_a'$, we can bound the norm of the difference as

$$\|\hat{\boldsymbol{S}}(\hat{\boldsymbol{\beta}}_{\boldsymbol{W}}) - \hat{\boldsymbol{S}}(\boldsymbol{\beta})\|_{\mathrm{F}} < 3ap^{3/2}\kappa\{\operatorname{tr}(\boldsymbol{\Omega}_2^{-1})\}^{1/2}.$$

This event happens with probability at least $1 - 3/a$. Thus, the result of the lemma follows by setting $a = 3/\varepsilon$.

LEMMA B.3. *Define matrix-valued function $\boldsymbol{S}$ analogously to $\hat{\boldsymbol{S}}$, replacing $\hat{\boldsymbol{\eta}}_i$ by $\boldsymbol{\eta}_i$ in the definition of $\boldsymbol{S}$. If Assumptions 1, 2, 3, and 5 are in force, then*

$$E\|\hat{\boldsymbol{S}}(\boldsymbol{\beta}) - \boldsymbol{S}(\boldsymbol{\beta}) - \phi \boldsymbol{B}\|_{\mathrm{F}}^2 \leq 9\kappa^2(\lambda+2)\operatorname{tr}(\boldsymbol{\Omega}_2^{-1}).$$

PROOF. Note that $\boldsymbol{V}_i^{\mathrm{T}}\boldsymbol{\eta}_i - \boldsymbol{V}_{i1}^{\mathrm{T}}\boldsymbol{\beta} = \boldsymbol{V}_{i2}^{\mathrm{T}}\boldsymbol{u}_i$ and

$$(\boldsymbol{V}_i^{\mathrm{T}}\hat{\boldsymbol{\eta}}_i - \boldsymbol{V}_{i1}^{\mathrm{T}}\boldsymbol{\beta})(\boldsymbol{V}_i^{\mathrm{T}}\hat{\boldsymbol{\eta}}_i - \boldsymbol{V}_{i1}^{\mathrm{T}}\boldsymbol{\beta})^{\mathrm{T}} = \boldsymbol{V}_{i2}^{\mathrm{T}}\boldsymbol{u}_i\boldsymbol{u}_i^{\mathrm{T}}\boldsymbol{V}_{i2} - \boldsymbol{V}_{i2}^{\mathrm{T}}\boldsymbol{u}_i\boldsymbol{h}_i^{\mathrm{T}}\boldsymbol{V}_i - \boldsymbol{V}_i^{\mathrm{T}}\boldsymbol{h}_i\boldsymbol{u}_i^{\mathrm{T}}\boldsymbol{V}_{i2} + \boldsymbol{V}_i^{\mathrm{T}}\boldsymbol{h}_i\boldsymbol{h}_i^{\mathrm{T}}\boldsymbol{V}_i.$$

From Assumptions 1 and 3 imply that

$$E(\boldsymbol{V}_{i2}^{\mathrm{T}}\boldsymbol{u}_i\boldsymbol{h}_i^{\mathrm{T}}\boldsymbol{V}_i) = 0,$$
$$E(\boldsymbol{V}_i^{\mathrm{T}}\boldsymbol{h}_i\boldsymbol{h}_i^{\mathrm{T}}\boldsymbol{V}_i) = \phi \boldsymbol{D}_i^{-2},$$
$$\operatorname{cov}\{\operatorname{vec}(\boldsymbol{V}_{i2}^{\mathrm{T}}\boldsymbol{u}_i\boldsymbol{h}_i^{\mathrm{T}}\boldsymbol{V}_i)\} = (\boldsymbol{V}_{i2}^{\mathrm{T}}\boldsymbol{\Sigma}\boldsymbol{V}_{i2}) \otimes (\phi \boldsymbol{D}_i^{-2}),$$
$$\operatorname{cov}\{\operatorname{vec}(\boldsymbol{D}_i\boldsymbol{V}_i^{\mathrm{T}}\boldsymbol{h}_i\boldsymbol{h}_i^{\mathrm{T}}\boldsymbol{V}_i\boldsymbol{D}_i)\} \preceq \lambda(\phi \boldsymbol{I}_{r_i}) \otimes (\phi \boldsymbol{I}_{r_i}).$$

Also, $\boldsymbol{W}_i(\phi \boldsymbol{D}_i^{-2})\boldsymbol{W}_i \preceq \kappa \boldsymbol{W}_i$ and $\boldsymbol{W}_i\boldsymbol{V}_{i2}^{\mathrm{T}}\boldsymbol{\Sigma}\boldsymbol{V}_{i2}\boldsymbol{W}_i \preceq \kappa \boldsymbol{W}_i$. Using the identities $E\|\boldsymbol{x}+\boldsymbol{y}+\boldsymbol{z}\|^2 \leq 9(E\|\boldsymbol{x}\|^2 + E\|\boldsymbol{y}\|^2 + E\|\boldsymbol{z}\|^2)$ and $E\|\boldsymbol{x}\|^2 = \operatorname{tr}\{\operatorname{cov}(\boldsymbol{x})\} + \|E(\boldsymbol{x})\|^2$, the result follows.

LEMMA B.4. *If Assumptions 1 and 2 are in force, then $E\{\boldsymbol{S}(\boldsymbol{\beta})\} = \boldsymbol{\Sigma}$ and*

$$E\|\boldsymbol{S}(\boldsymbol{\beta}) - \boldsymbol{\Sigma}\|_{\mathrm{F}}^2 \leq \mu\kappa^2 \operatorname{tr}(\boldsymbol{\Omega}_2^{-1}).$$

PROOF. Write

$$\boldsymbol{A}(\boldsymbol{\beta}) = \sum_{i=1}^{M} \boldsymbol{V}_{i2}\boldsymbol{W}_i\boldsymbol{V}_{i2}^{\mathrm{T}}\boldsymbol{\Sigma}^{1/2}\tilde{\boldsymbol{u}}_i\tilde{\boldsymbol{u}}_i^{\mathrm{T}}\boldsymbol{\Sigma}^{1/2}\boldsymbol{V}_{i2}\boldsymbol{W}_i\boldsymbol{V}_{i2}^{\mathrm{T}},$$

Since $E(\tilde{\boldsymbol{u}}_i\tilde{\boldsymbol{u}}_i^{\mathrm{T}}) = \boldsymbol{I}_q$, it follows that $E\{\boldsymbol{S}(\boldsymbol{\beta})\} = \boldsymbol{\Sigma}$. Note that

$$\operatorname{cov}\{\operatorname{vec}(\tilde{\boldsymbol{u}}_i\tilde{\boldsymbol{u}}_i^{\mathrm{T}})\} \preceq E\|\tilde{\boldsymbol{u}}_i\|^4 \boldsymbol{I} \preceq \mu \boldsymbol{I}.$$

Using this relation and the fact that $\boldsymbol{W}_i\boldsymbol{V}_{i2}^{\mathrm{T}}\boldsymbol{\Sigma}\boldsymbol{V}_{i2}\boldsymbol{W}_i \preceq \kappa \boldsymbol{W}_i$ it follows that

$$\operatorname{cov}[\operatorname{vec}\{\boldsymbol{S}(\boldsymbol{\beta})\}] \preceq \mu \boldsymbol{\Omega}_2^{-1}\Big\{\sum_{i=1}^{M}(\boldsymbol{V}_{i2}\boldsymbol{W}_i\boldsymbol{V}_{i2}^{\mathrm{T}}\boldsymbol{\Sigma}\boldsymbol{V}_{i2}\boldsymbol{W}_i\boldsymbol{V}_{i2}^{\mathrm{T}}) \otimes (\boldsymbol{V}_{i2}\boldsymbol{W}_i\boldsymbol{V}_{i2}^{\mathrm{T}}\boldsymbol{\Sigma}\boldsymbol{V}_{i2}\boldsymbol{W}_i\boldsymbol{V}_{i2}^{\mathrm{T}})\Big\}\boldsymbol{\Omega}_2^{-1}$$
$$\preceq \mu\kappa^2 \boldsymbol{\Omega}_2^{-1}\Big\{\sum_{i=1}^{M}(\boldsymbol{V}_{i2}\boldsymbol{W}_i\boldsymbol{V}_{i2}^{\mathrm{T}}) \otimes (\boldsymbol{V}_{i2}\boldsymbol{W}_i\boldsymbol{V}_{i2}^{\mathrm{T}})\Big\}\boldsymbol{\Omega}_2^{-1}$$
$$= \mu\kappa^2 \boldsymbol{\Omega}_2^{-1}.$$

The result of the lemma follows.

LEMMA B.5. *If Assumption 5 is in force, then*

$$\|\boldsymbol{B}\|_{\mathrm{F}} \leq \phi^{-1}\kappa\rho^{1/2}\|\boldsymbol{\Omega}_2^{-1}\|^{1/2},$$

*where $\rho = \sum_{i=1}^{M} r_i$.*

PROOF. Write

$$\operatorname{vec}(\boldsymbol{B}) = \sum_{i=1}^{M} \boldsymbol{A}_i \operatorname{vec}(\boldsymbol{I}_{r_i}),$$

where $\boldsymbol{A}_i = \boldsymbol{\Omega}_2^{-1}(\boldsymbol{V}_{i2}\boldsymbol{W}_i\boldsymbol{D}_i^{-1} \otimes \boldsymbol{V}_{i2}\boldsymbol{W}_i\boldsymbol{D}_i^{-1})$. From Assumption 5, it follows that

$$\sum_{i=1}^{M} \boldsymbol{A}_i^{\mathrm{T}} \boldsymbol{A}_i \preceq \phi^{-2} \kappa^2 \boldsymbol{\Omega}_2^{-1}.$$

From Lemma B.1, it follows that $\|\operatorname{vec}(\boldsymbol{B})\|^2 \leq \phi^{-2}\kappa^2 \|\boldsymbol{\Omega}_2^{-1}\| \sum_{i=1}^{M} \|\operatorname{vec}(\boldsymbol{I}_{r_i})\|^2 = \phi^{-2}\kappa^2\rho\|\boldsymbol{\Omega}_2^{-1}\|$. This gives the result of the lemma since $\|\boldsymbol{B}\|_{\mathrm{F}} = \|\operatorname{vec}(\boldsymbol{B})\|$.

## C. Proof of Lemma 5.9

We prove Lemma 5.9 in a series of smaller lemmas, labeled Lemma C.2–C.4. These results rely on a set of bounds derived as part of the following result.

LEMMA C.1. *Consider the weight matrices $\boldsymbol{W}_{\boldsymbol{\theta} i}$ and $\boldsymbol{\Omega}_{\boldsymbol{\theta}}$ defined in (16), with $\boldsymbol{W}_i$ and $\boldsymbol{\Omega}$ defined by setting $\boldsymbol{\theta} = \boldsymbol{\theta}_0$. Set $\boldsymbol{E}_{\boldsymbol{\theta} i} = \boldsymbol{V}_{i2}^{\mathrm{T}}(\bar{\boldsymbol{\Sigma}}_{\boldsymbol{\theta}} - \bar{\boldsymbol{\Sigma}})\boldsymbol{V}_{i2}$. The following identities hold:*

$$\boldsymbol{W}_{\boldsymbol{\theta} i} - \boldsymbol{W}_i = -\boldsymbol{W}_{\boldsymbol{\theta} i}\boldsymbol{E}_{\boldsymbol{\theta} i}\boldsymbol{W}_i, \tag{1a}$$

$$\boldsymbol{\Omega}_{\boldsymbol{\theta}}^{-1} - \boldsymbol{\Omega}^{-1} = \boldsymbol{\Omega}_{\boldsymbol{\theta}}^{-1}\Big(\sum_{i=1}^{M} \boldsymbol{W}_{\boldsymbol{\theta} i}\boldsymbol{E}_{\boldsymbol{\theta} i}\boldsymbol{W}_i\Big)\boldsymbol{\Omega}^{-1}, \tag{1b}$$

*along with the inequalities*

$$\|\boldsymbol{W}_i\boldsymbol{E}_{\boldsymbol{\theta} i}\| \leq 4\|\bar{\boldsymbol{\Sigma}}^{-1}\bar{\boldsymbol{\Sigma}}_{\boldsymbol{\theta}} - \boldsymbol{I}_q\|, \tag{2a}$$

$$\|\boldsymbol{\Omega}_{\boldsymbol{\theta}}^{-1}\boldsymbol{\Omega} - \boldsymbol{I}_p\| \leq 4p\|\bar{\boldsymbol{\Sigma}}^{-1}\bar{\boldsymbol{\Sigma}}_{\boldsymbol{\theta}} - \boldsymbol{I}_q\|. \tag{2b}$$

PROOF. Equation (1a) follows from the identity $(\boldsymbol{A} + \boldsymbol{E})^{-1} - \boldsymbol{A}^{-1} = -(\boldsymbol{A} + \boldsymbol{E})^{-1}\boldsymbol{E}\boldsymbol{A}^{-1}$, which holds whenever $\boldsymbol{A}$ and $(\boldsymbol{A}+\boldsymbol{E})$ are invertible. Applying the identity again to $\boldsymbol{\Omega}_{\boldsymbol{\theta}}$ gives (1b).

To show (2a), we use that if $\boldsymbol{A}$ and $\boldsymbol{B}$ are symmetric, $\boldsymbol{A}$ is positive-definite, and $\boldsymbol{B}$ is positive semidefinite, then the following holds:

$$\begin{aligned}
\|(\boldsymbol{V}^{\mathrm{T}}\boldsymbol{A}\boldsymbol{V} + \boldsymbol{D}^{-2})^{-1}\boldsymbol{V}^{\mathrm{T}}\boldsymbol{B}\boldsymbol{V}\| &= \|(\boldsymbol{V}^{\mathrm{T}}\boldsymbol{A}\boldsymbol{V} + \boldsymbol{D}^{-2})^{-1}\boldsymbol{V}^{\mathrm{T}}\boldsymbol{A}^{1/2}(\boldsymbol{A}^{-1/2}\boldsymbol{B}\boldsymbol{A}^{-1/2})\boldsymbol{A}^{1/2}\boldsymbol{V}^{\mathrm{T}}\| \\
&= \|(\boldsymbol{V}^{\mathrm{T}}\boldsymbol{A}\boldsymbol{V} + \boldsymbol{D}^{-2})^{-1/2}\boldsymbol{V}^{\mathrm{T}}\boldsymbol{A}^{1/2}(\boldsymbol{A}^{-1/2}\boldsymbol{B}\boldsymbol{A}^{-1/2})^{1/2}\|^2 \\
&\leq \|\boldsymbol{A}^{-1}\boldsymbol{B}\|\|(\boldsymbol{V}^{\mathrm{T}}\boldsymbol{A}\boldsymbol{V} + \boldsymbol{D}^{-2})^{-1}\boldsymbol{V}^{\mathrm{T}}\boldsymbol{A}\boldsymbol{V}\| \\
&= \|\boldsymbol{A}^{-1}\boldsymbol{B}\|\|\boldsymbol{I} - (\boldsymbol{V}^{\mathrm{T}}\boldsymbol{A}\boldsymbol{V} + \boldsymbol{D}^{-2})^{-1}\boldsymbol{D}^{-2}\| \\
&\leq 2\|\boldsymbol{A}^{-1}\boldsymbol{B}\|.
\end{aligned}$$

If $\boldsymbol{B}$ is symmetric but not positive semidefinite, then the matrix $\boldsymbol{C} = \boldsymbol{A}^{-1/2}\boldsymbol{B}\boldsymbol{A}^{-1/2}$ might not be positive semidefinite, so $\boldsymbol{C}^{1/2}$ may not exist. In this case, $\boldsymbol{C}$ can be written as $\boldsymbol{C} = \boldsymbol{C}_+ - \boldsymbol{C}_-$, where $\boldsymbol{C}_+$ and $\boldsymbol{C}_-$ are positive semidefinite. Then, since $\max\|\{\|\boldsymbol{C}_+\|, \|\boldsymbol{C}_-\|\} \leq \|\boldsymbol{C}\|$ it follows that

$$\|(\boldsymbol{V}^{\mathrm{T}}\boldsymbol{A}\boldsymbol{V} + \boldsymbol{D}^{-2})^{-1}\boldsymbol{V}^{\mathrm{T}}\boldsymbol{B}\boldsymbol{V}\| \leq 2(\|\boldsymbol{C}_+\| + \|\boldsymbol{C}_-\|) \leq 4\|\boldsymbol{C}\| = 4\|\boldsymbol{A}^{-1}\boldsymbol{B}\|.$$

The last inequality (2b) follows from (1b) and (2a) using the bound

$$\Big\| \sum_{i=1}^{M} \boldsymbol{\Omega}_{\boldsymbol{\theta}}^{-1} \boldsymbol{W}_{\boldsymbol{\theta} i} \boldsymbol{E}_{\boldsymbol{\theta} i} \boldsymbol{W}_i \Big\| \leq \max_i \|\boldsymbol{E}_{\boldsymbol{\theta} i} \boldsymbol{W}_i\| \sum_{i=1}^{M} \|\boldsymbol{\Omega}_{\boldsymbol{\theta}}^{-1} \boldsymbol{W}_{\boldsymbol{\theta} i}\|$$

$$\leq \max_i \|\boldsymbol{E}_{\boldsymbol{\theta} i} \boldsymbol{W}_i\| \sum_{i=1}^{M} \operatorname{tr}(\boldsymbol{\Omega}_{\boldsymbol{\theta}}^{-1} \boldsymbol{W}_{\boldsymbol{\theta} i})$$

$$= \max_i \|\boldsymbol{E}_{\boldsymbol{\theta} i} \boldsymbol{W}_i\| \operatorname{tr}(\boldsymbol{I}_p).$$

LEMMA C.2. *Let function $\boldsymbol{\delta}_1(\boldsymbol{\theta})$ be defined as in (17a). If Assumptions 1–3 are in force, then for any $\varepsilon > 0$ and any parameter set $\mathcal{B}$,*

$$\Pr\{\sup_{\boldsymbol{\theta} \in \mathcal{B}} \|\boldsymbol{\delta}_1(\boldsymbol{\theta})\|^2 \geq 16\varepsilon^{-1} p^3 \tau^2\} \leq \varepsilon.$$

*where $\tau = \sup_{\boldsymbol{\theta} \in \mathcal{B}} \|\bar{\boldsymbol{\Sigma}}^{-1} \bar{\boldsymbol{\Sigma}}_{\boldsymbol{\theta}} - \boldsymbol{I}_q\|$.*

PROOF. Set $\boldsymbol{m} = \boldsymbol{\Omega}^{-1/2} \sum_{i=1}^{M} \boldsymbol{V}_{i1}^{\mathrm{T}} \boldsymbol{W}_i^{1/2} \boldsymbol{\gamma}_i$, noting that $E(\boldsymbol{m}) = 0$ and $\operatorname{cov}(\boldsymbol{m}) = \boldsymbol{I}_p$. It follows that $E\|\boldsymbol{m}\|^2 = p$, so by Markov's inequality, $\Pr(\|\boldsymbol{m}\|^2 \geq \varepsilon^{-1} p) \leq \varepsilon$. Since

$$\|\boldsymbol{\delta}_1(\boldsymbol{\theta})\| = \|(\boldsymbol{\Omega}^{1/2} \boldsymbol{\Omega}_{\boldsymbol{\theta}}^{-1} \boldsymbol{\Omega}^{1/2} - \boldsymbol{I}_p)\, \boldsymbol{m}\| \leq \|\boldsymbol{\Omega}_{\boldsymbol{\theta}}^{-1} \boldsymbol{\Omega} - \boldsymbol{I}_p\| \|\boldsymbol{m}\|,$$

the bound follows from Lemma C.1.

LEMMA C.3. *Let $\boldsymbol{\delta}_2(\boldsymbol{\theta})$ be defined as in (17b). If Assumptions 1–3 are in force, then for any $\varepsilon > 0$ and any parameter set $\mathcal{B}$,*

$$\Pr\{\sup_{\boldsymbol{\theta} \in \mathcal{B}} \|\boldsymbol{\delta}_2(\boldsymbol{\theta})\|^2 \geq 16\varepsilon^{-1} pq^3 \tau^2 (1 + 4p\tau)^2\} \leq \varepsilon.$$

*where $\tau = \sup_{\boldsymbol{\theta} \in \mathcal{B}} \|\bar{\boldsymbol{\Sigma}}^{-1} \bar{\boldsymbol{\Sigma}}_{\boldsymbol{\theta}} - \boldsymbol{I}_q\|$.*

PROOF. Define the function

$$\boldsymbol{m}(\boldsymbol{\theta}) = \boldsymbol{\Omega}^{-1/2} \sum_{i=1}^{M} \boldsymbol{V}_{i1} \boldsymbol{W}_i \boldsymbol{E}_{\boldsymbol{\theta} i} \boldsymbol{W}_i^{1/2} \boldsymbol{\gamma}_i.$$

Note that $\boldsymbol{m}(\boldsymbol{\theta}_0) = 0$. Also, for any fixed $\boldsymbol{\theta}$, $E\{\boldsymbol{m}(\boldsymbol{\theta})\} = 0$ and by Lemma C.1,

$$\operatorname{cov}\{\boldsymbol{m}(\boldsymbol{\theta})\} = \boldsymbol{\Omega}^{-1/2} \Big\{ \sum_{i=1}^{M} \boldsymbol{V}_{i1} \boldsymbol{W}_i \boldsymbol{E}_{\boldsymbol{\theta} i} \boldsymbol{W}_i \boldsymbol{W}_i \boldsymbol{E}_{\boldsymbol{\theta} i} \boldsymbol{V}_{i1}^{\mathrm{T}} \Big\} \boldsymbol{\Omega}^{-1/2} \preceq 16 \|\bar{\boldsymbol{\Sigma}}^{-1} \bar{\boldsymbol{\Sigma}}_{\boldsymbol{\theta}} - \boldsymbol{I}_q\|^2 \boldsymbol{I}_p.$$

Thus, Markov's inequality implies that for any fixed $a > 0$ and any fixed $\boldsymbol{\theta}$,

$$\Pr\{\|\boldsymbol{m}(\boldsymbol{\theta})\|^2 \geq a\} \leq 16 a^{-1} p \|\bar{\boldsymbol{\Sigma}}^{-1} \bar{\boldsymbol{\Sigma}}_{\boldsymbol{\theta}} - \boldsymbol{I}_q\|^2. \qquad (3)$$

We will use this pointwise bound and the fact that $\boldsymbol{m}(\boldsymbol{\theta})$ is linear in $\bar{\boldsymbol{\Sigma}}_{\boldsymbol{\theta}}$ to bound the supremum.

For $k = 1, \ldots, q$, define $\bar{\boldsymbol{\Sigma}}_{kk} = \bar{\boldsymbol{\Sigma}}^{1/2} (\boldsymbol{I} + \boldsymbol{e}_k \boldsymbol{e}_k^{\mathrm{T}}) \bar{\boldsymbol{\Sigma}}^{1/2}$, where $\boldsymbol{e}_k$ denotes the $k$th standard basis vector. Similarly, for $l \neq k$, define $\bar{\boldsymbol{\Sigma}}_{kl} = \bar{\boldsymbol{\Sigma}}^{1/2} (\boldsymbol{I} + \boldsymbol{e}_k \boldsymbol{e}_l^{\mathrm{T}} + \boldsymbol{e}_l \boldsymbol{e}_k^{\mathrm{T}}) \bar{\boldsymbol{\Sigma}}^{1/2}$. For $1 \leq k \leq l \leq q$ choose a $\boldsymbol{\theta}_{kl}$ such that $\bar{\boldsymbol{\Sigma}}_{\boldsymbol{\theta}_{kl}} = \bar{\boldsymbol{\Sigma}}$. For any parameter vector $\boldsymbol{\theta}$, define

$$t_{kl}(\boldsymbol{\theta}) = \boldsymbol{e}_k^{\mathrm{T}} (\bar{\boldsymbol{\Sigma}}^{-1/2} \bar{\boldsymbol{\Sigma}}_{\boldsymbol{\theta}} \bar{\boldsymbol{\Sigma}}^{-1/2} - \boldsymbol{I}_q) \boldsymbol{e}_l;$$

it follows that
$$\bar{\boldsymbol{\Sigma}}_{\boldsymbol{\theta}} = \bar{\boldsymbol{\Sigma}} + \sum_{k \leq l} t_{kl}(\boldsymbol{\theta})(\bar{\boldsymbol{\Sigma}}_{kl} - \bar{\boldsymbol{\Sigma}}).$$

In particular, since $\boldsymbol{m}(\boldsymbol{\theta})$ is linear in $\bar{\boldsymbol{\Sigma}}_{\boldsymbol{\theta}}$, this implies that
$$\boldsymbol{m}(\boldsymbol{\theta}) = \sum_{k \leq l} t_{kl}(\boldsymbol{\theta}) \boldsymbol{m}(\boldsymbol{\theta}_{kl}).$$

Thus, the matrix Cauchy-Schwarz inequality (Lemma B.1) implies that
$$\|\boldsymbol{m}(\boldsymbol{\theta})\|^2 \leq \Big[\sum_{k \leq l} \{t_{kl}(\boldsymbol{\theta})\}^2\Big]\Big[\sum_{k \leq l} \|\boldsymbol{m}(\boldsymbol{\theta}_{kl})\|^2\Big].$$

The first term on the right hand side is bounded by $\|\bar{\boldsymbol{\Sigma}}^{-1}\bar{\boldsymbol{\Sigma}}_{\boldsymbol{\theta}} - \boldsymbol{I}_q\|_{\text{F}}^2$ which in turn is bounded by $q\|\bar{\boldsymbol{\Sigma}}^{-1}\bar{\boldsymbol{\Sigma}}_{\boldsymbol{\theta}} - \boldsymbol{I}_q\|^2$. The second term is bounded by $q^2 \max_{k \leq l} \|\boldsymbol{m}(\boldsymbol{\theta}_{kl})\|^2$.

From (3), for any $a > 0$,
$$\Pr\{\max_{k \leq l} \|\boldsymbol{m}(\boldsymbol{\theta}_{kl})\|^2 \geq a\} \leq 16 a^{-1} pq^2.$$

Thus, if we set $a = 16\varepsilon^{-1} pq^2$, then
$$\Pr\{\sup_{\boldsymbol{\theta} \in \mathcal{B}} \|\boldsymbol{m}(\boldsymbol{\theta})\|^2 \geq 16\varepsilon^{-1} pq^3 \sup_{\boldsymbol{\theta} \in \mathcal{B}} \|\bar{\boldsymbol{\Sigma}}^{-1}\bar{\boldsymbol{\Sigma}}_{\boldsymbol{\theta}} - \boldsymbol{I}_q\|^2\} \leq \varepsilon.$$

Finally, Lemma C.1 and the triangle inequality imply that
$$\|\boldsymbol{\Omega}^{1/2}\boldsymbol{\Omega}_{\boldsymbol{\theta}}^{-1}\boldsymbol{\Omega}^{1/2}\| = \|\boldsymbol{\Omega}_{\boldsymbol{\theta}}^{-1}\boldsymbol{\Omega}\| \leq 1 + 4p\|\bar{\boldsymbol{\Sigma}}^{-1}\bar{\boldsymbol{\Sigma}}_{\boldsymbol{\theta}} - \boldsymbol{I}_q\|.$$

Putting these two bounds together gives the result of the lemma.

LEMMA C.4. *Let function $\boldsymbol{\delta}_3(\boldsymbol{\theta})$ be defined as in (17c). If Assumptions 1–3 are in force, then for any $\varepsilon > 0$ and any parameter set $\mathcal{B}$,*
$$\Pr\{\sup_{\boldsymbol{\theta} \in \mathcal{B}} \|\boldsymbol{\delta}_3(\boldsymbol{\theta})\|^2 \geq 256\varepsilon^{-1} \rho \tau^4 (1 + 4p\tau)\} \leq \varepsilon,$$

*where $\tau = \sup_{\boldsymbol{\theta} \in \mathcal{B}} \max\{\|\bar{\boldsymbol{\Sigma}}^{-1}\bar{\boldsymbol{\Sigma}}_{\boldsymbol{\theta}} - \boldsymbol{I}_p\|, \|\bar{\boldsymbol{\Sigma}}_{\boldsymbol{\theta}}^{-1}\bar{\boldsymbol{\Sigma}} - \boldsymbol{I}_p\|\}$ and $\rho = \sum_{i=1}^M r_i$.*

PROOF. Define
$$\boldsymbol{m}(\boldsymbol{\theta}) = \boldsymbol{\Omega}_{\boldsymbol{\theta}}^{-1/2} \sum_{i=1}^M \boldsymbol{V}_{i1} \boldsymbol{W}_{\boldsymbol{\theta} i} \boldsymbol{E}_{\boldsymbol{\theta} i} \boldsymbol{W}_i \boldsymbol{E}_{\boldsymbol{\theta} i} \boldsymbol{W}_i^{1/2} \boldsymbol{\gamma}_i.$$

Applying Cauchy-Schwarz (Lemma B.1) and Lemma C.1, we get
$$\|\boldsymbol{m}(\boldsymbol{\theta})\|^2 \leq \sum_{i=1}^M \boldsymbol{\gamma}_i^{\text{T}} \boldsymbol{W}_i^{1/2} \boldsymbol{E}_{\boldsymbol{\theta} i} \boldsymbol{W}_i \boldsymbol{E}_{\boldsymbol{\theta} i} \boldsymbol{W}_{\boldsymbol{\theta} i} \boldsymbol{E}_{\boldsymbol{\theta} i} \boldsymbol{W}_i \boldsymbol{E}_{\boldsymbol{\theta} i} \boldsymbol{W}_i^{1/2} \boldsymbol{\gamma}_i$$
$$\leq \max_i \|\boldsymbol{W}_i \boldsymbol{E}_{\boldsymbol{\theta} i}\|^2 \max_i \|\boldsymbol{W}_i^{1/2} \boldsymbol{E}_{\boldsymbol{\theta} i} \boldsymbol{W}_{\boldsymbol{\theta} i}^{1/2}\|^2 \sum_{i=1}^M \|\boldsymbol{\gamma}_i\|^2.$$
$$\leq 256 \|\bar{\boldsymbol{\Sigma}}^{-1}\bar{\boldsymbol{\Sigma}}_{\boldsymbol{\theta}} - \boldsymbol{I}_q\|^3 \|\bar{\boldsymbol{\Sigma}}_{\boldsymbol{\theta}}^{-1}\bar{\boldsymbol{\Sigma}} - \boldsymbol{I}_q\| \sum_{i=1}^M \|\boldsymbol{\gamma}_i\|^2.$$

Since $E(\sum_{i=1}^{M}\|\boldsymbol{\gamma}_i\|^2) = \rho$, Markov's inequality gives that

$$\Pr\{\sum_{i=1}^{M}\|\boldsymbol{\gamma}_i\|^2 \geq \varepsilon^{-1}\rho\} \leq \varepsilon.$$

Also, Lemma C.1 gives that $\|\boldsymbol{\Omega}^{1/2}\boldsymbol{\Omega}_{\boldsymbol{\theta}}^{-1/2}\|^2 = \|\boldsymbol{\Omega}\boldsymbol{\Omega}_{\boldsymbol{\theta}}^{-1}\| \leq (1 + 4p\|\bar{\boldsymbol{\Sigma}}^{-1}\bar{\boldsymbol{\Sigma}}_{\boldsymbol{\theta}} - \boldsymbol{I}_p\|)$. This gives the result of the lemma since $\boldsymbol{\delta}_3(\boldsymbol{\theta}) = \boldsymbol{\Omega}^{1/2}\boldsymbol{\Omega}_{\boldsymbol{\theta}}^{-1/2}\boldsymbol{m}(\boldsymbol{\theta})$.

### D. Proof of Proposition 6.2

Take $\tau_N$ to be any sequence of positive real numbers, and define the parameter set $\mathcal{B}_N = \{\boldsymbol{\theta} = (\phi_{\boldsymbol{\theta}}, \boldsymbol{\Sigma}_{\boldsymbol{\theta}}) : \|\bar{\boldsymbol{\Sigma}}^{-1}\bar{\boldsymbol{\Sigma}}_{\boldsymbol{\theta}} - \boldsymbol{I}_q\| \leq \tau_N \text{ and } \|\bar{\boldsymbol{\Sigma}}_{\boldsymbol{\theta}}^{-1}\bar{\boldsymbol{\Sigma}} - \boldsymbol{I}_q\| \leq \tau_N\}$, where $\bar{\boldsymbol{\Sigma}}_{\boldsymbol{\theta}} = \phi_{\boldsymbol{\theta}}^{-1}\boldsymbol{\Sigma}_{\boldsymbol{\theta}}$. Assumption 4 and Proposition 5.6 imply that if $(N - \rho_N)^{1/2}\tau_N \to \infty$ and $\{1 + (N/\rho_N - 1)^{-1/2}\}\omega_{N,2}^{1/2}\tau_N \to \infty$, then $\Pr(\hat{\boldsymbol{\theta}} \in \mathcal{B}_N) \to 1$. If, in addition, $\tau_N \to 0$ and $\rho_N\tau_N^4 \to 0$, then Proposition 5.8 implies that $\boldsymbol{\Omega}^{1/2}(\hat{\boldsymbol{\beta}}_{\hat{\boldsymbol{\theta}}} - \hat{\boldsymbol{\beta}}) \xrightarrow{p} 0$. Setting $\tau_N^4 = \rho_N^{-1}\log\rho_N$ is sufficient to guarantee that these conditions hold.

### E. Stochastic gradient descent for $h$-likelihood

Reparameterizing the random effect covariance matrix as $\boldsymbol{\Sigma} = \sigma^2\boldsymbol{\Lambda}^{-1}$, the logarithm of the $h$-likelihood for the normal hierarchical model can be written as

$$h(\boldsymbol{\beta}, \boldsymbol{\Lambda}, \sigma^2, \boldsymbol{u} \mid \boldsymbol{y}) = \sum_{i=1}^{M}\sum_{j=1}^{n_i} h_{ij},$$

where

$$h_{ij} = -\log\sigma - \frac{1}{2\sigma^2}(y_{ij} - \boldsymbol{x}_{ij}^{\mathrm{T}}\boldsymbol{\beta} - \boldsymbol{z}_{ij}^{\mathrm{T}}\boldsymbol{u}_i)^2 - \frac{1}{2n_i}\log|\sigma^2\boldsymbol{\Lambda}^{-1}| - \frac{1}{2\sigma^2 n_i}\boldsymbol{u}_i^{\mathrm{T}}\boldsymbol{\Lambda}\boldsymbol{u}_i.$$

For fixed $\boldsymbol{\Lambda}$, let $\boldsymbol{\beta}_{\boldsymbol{\Lambda}}$ and $\boldsymbol{u}_{\boldsymbol{\Lambda}}$ denote the values of $\boldsymbol{\beta}$ and $\boldsymbol{u}$ that maximize $h$; note that these values do note depend on $\sigma^2$.

If we set $\sigma^2 = 1$, then the gradients of the summands with respect to $\boldsymbol{\beta}$ and $\boldsymbol{u}_i$ are

$$\nabla_{\boldsymbol{\beta}} h_{ij} = (y_{ij} - \boldsymbol{x}_{ij}^{\mathrm{T}}\boldsymbol{\beta} - \boldsymbol{z}_{ij}^{\mathrm{T}}\boldsymbol{u}_i)\,\boldsymbol{x}_{ij},$$
$$\nabla_{\boldsymbol{u}_i} h_{ij} = (y_{ij} - \boldsymbol{x}_{ij}^{\mathrm{T}}\boldsymbol{\beta} - \boldsymbol{z}_{ij}^{\mathrm{T}}\boldsymbol{u}_i)\,\boldsymbol{z}_{ij} - (1/n_i)\boldsymbol{\Lambda}\boldsymbol{u}_i.$$

To apply a variant of the Robbins and Monro (1951) stochastic gradient descent method to find $\boldsymbol{\beta}_{\boldsymbol{\Lambda}}$ and $\boldsymbol{u}_{\boldsymbol{\Lambda}}$, we start with randomly chosen initial values $\boldsymbol{\beta}^0$ and $\boldsymbol{u}^0$. At each iteration, we perform $N$ steps, processing all observations in random order. In the $t$th step, when we process observation $ij$, we choose a learning rate $\alpha_t$ and perform the updates

$$\boldsymbol{\beta}^{t+1} = \boldsymbol{\beta}^t + \alpha_t\,(y_{ij} - \boldsymbol{x}_{ij}^{\mathrm{T}}\boldsymbol{\beta}^t - \boldsymbol{z}_{ij}^{\mathrm{T}}\boldsymbol{u}_i^t)\,\boldsymbol{x}_{ij},$$
$$\boldsymbol{u}_i^{k+1} = \boldsymbol{u}_i^k + \alpha_t\{(y_{ij} - \boldsymbol{x}_{ij}^{\mathrm{T}}\boldsymbol{\beta}^t - \boldsymbol{z}_{ij}^{\mathrm{T}}\boldsymbol{u}_i^t)\,\boldsymbol{z}_{ij}, -(1/n_i)\boldsymbol{\Lambda}\boldsymbol{u}_i^t\};$$

for $i' \neq i$, we retain the value of the random effect vector $\boldsymbol{u}_{i'}^t$. Following Darken and Moody (1991), it is common to take the step size as

$$\alpha_t = \alpha_0(1 + at)^{-1}$$

for appropriately chosen values $\alpha_0$ and $a$. Following the recommendations of Y. LeCun and Müller (1998) and Bottou (2012), we use different learning rates for the fixed effect vector and for each random effect vector. For the $i$th random effect, vector, we set $a = \alpha_0(\|\mathbf{\Lambda}^{-1}\|^{-1}/n_i)$. For the fixed effects, we set $a = \alpha_0/N$.

In our simulations, we set $\alpha_0 = 0.1/N$ for the fixed effect vector and $\alpha_0 = 0.1/n_i$ for the $i$th random effect vector works well. We perform 30 iterations for fitting a linear model, and 100 iterations for fitting a logistic model.

### F. Performance in linear regression simulations

We perform a simulation analogous to the one described in Section 7, but with a hierarchical linear regression model instead of a hierarchical logistic model.

As in Section 7, we simulate $N$ samples drawn from $M$ groups, with $M = 1000$ and $N$ ranging from 100 to 100000. We set the dimensions of the fixed and random effect vectors to $p = q = 5$. For each choice of $N$, we perform 100 replicates.

For each replicate, we generate the $p$-dimensional fixed effect, $\beta$, and $q \times q$ random effect covariance matrix, $\mathbf{\Sigma}$ as before, but with a different scaling for the fixed effect. Specifically, for $k = 1, \ldots, p$, we draw
$beta_k$ independently from a $t$ distribution with 4 degrees of freedom. We draw coefficient matrix $\mathbf{\Sigma}$ from an inverse Wishart distribution with shape $\mathbf{I}$ and $2q$ degrees of freedom.

We draw the group specific sample sizes $n_i$ $(i = 1, \ldots, M)$, random effects $\mathbf{u}_i (i = 1, \ldots, M)$, and predictor vectors $\mathbf{x}_{ij}, \mathbf{z}_{ij}$ $(i = 1, \ldots, M; j = 1, \ldots, n_i)$ as in Section 7. To generate response value $y_{ij}$, we draw from a normal distribution with mean $\mu_{ij} = \mathbf{x}_{ij}^{\mathrm{T}}\boldsymbol{\beta} + \mathbf{z}_{ij}^{\mathrm{T}}\mathbf{u}_i$ and variance $\phi = 1$.

We compute moment-based estimates of the population parameters $\boldsymbol{\beta}$ and $\mathbf{\Sigma}$, along with approximate empirical Bayes estimates $\hat{\mathbf{u}}_i$ $(i = 1, \ldots, M)$ using the procedure described in Section 3.4. We compare the following methods:

(a) *mhglm*, the proposed moment-based estimation procedure. To compute the moment-based estimates, we use two-step estimators with semi-weighted initial step and $\mathbf{\Sigma}_0$ set to the identity matrix, after standardizing the predictors. The procedure is implemented in the R programming language.
(b) *lmer*, maximum likelihood, using a gradient-free optimization procedure applied to the profiled likelihood, implemented in C++ and R by the *lme4* R package (Bates et al., 2013).
(c) *sgd*, which uses stochastic gradient descent to maximize a regularized version of the *h*-likelihood. We implemented the compute-intensive inner loop in C, and the outer loop in R.
(d) *lmer split*, a data-splitting estimation procedure, which splits the data set into 10 subsets, computes separate estimates for each using *lmer*, and then combines the estimates by averaging them. Implemented in R.
(e) *lme*, maximum likelihood, using a combination of Expectation-Maximization and Newton Raphson iteration, implemented in the C and R programming languages by the *nlme* package (Pinheiro et al., 2014).

As in Section 3.4, we report serial computation time for each procedure, and we do not include cross-validation time for the tuning parameter selection for the *sgd* method.

We use the same fixed effect, random effect covariance, and random effect loss as in Section 7.

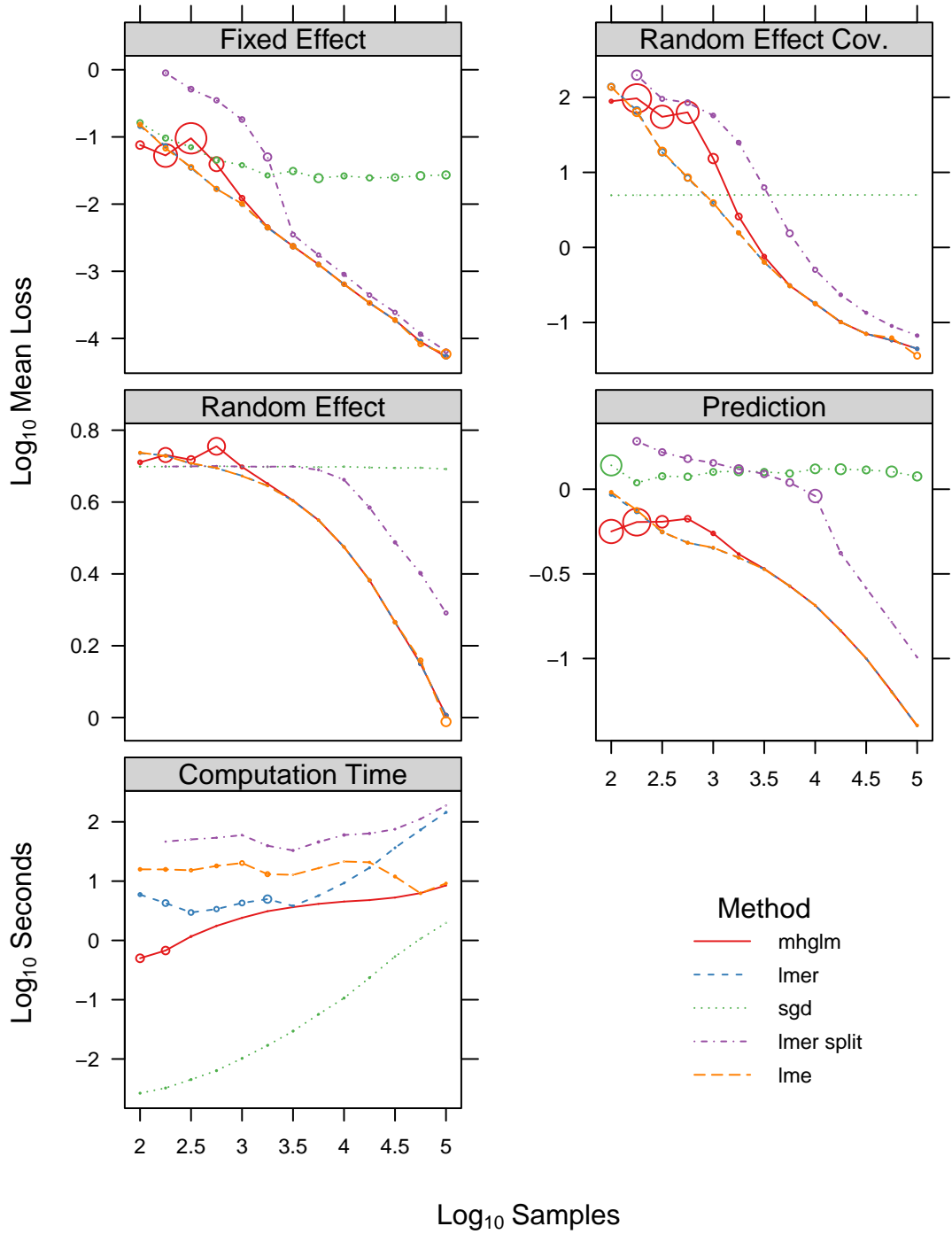

**Fig. 1.** Performance for the hierarchical linear model. Circle radii indicate one standard error along $y$-axis (absent when smaller than line width).

For prediction loss, we use the squared error loss

$$N^{-1} \sum_{i=1}^{M} \sum_{j=1}^{n_i} \phi^{-1}(\mu_{ij} - \hat{\mu}_{ij})^2$$

where $\mu_{ij} = \boldsymbol{x}_{ij}^{\mathrm{T}}\boldsymbol{\beta} + \boldsymbol{z}_{ij}^{\mathrm{T}}\boldsymbol{u}_i$ and $\hat{\mu}_{ij} = \boldsymbol{x}_{ij}^{\mathrm{T}}\hat{\boldsymbol{\beta}} + \boldsymbol{z}_{ij}^{\mathrm{T}}\hat{\boldsymbol{u}}_i$.

Fig. 1 shows the mean loss, averaged over all replicates, with circle radii indicating standard errors along the vertical axes (absent when less than the visible line width). For small sample sizes, the methods based on maximum likelihood (*lme* and *lmer*) perform slightly better than the proposed method (*mhglm*), but for large sample sizes, these methods all have similar loss performances. The *lme*, *lmer*, *lmer.split*, and *mhglm* methods all appear to give consistent estimates, with the *lmer.split* being less statistically efficient than that other three methods. The *sgd* method does not see a noticeable improvement in loss performance over the range of sample sizes considered.

The lower-left panel of Fig. 1 shows the sequential computation times for all methods. In terms of computation time, for most values of $N$ in the simulation, the proposed method performs better than the likelihood-based procedures. However, at $N = 10^{4.75}$ and $N = 10^5$, the proposed method has the same average running time as the *lme* method.

In terms of computation time, the *sgd* procedure is clearly the fastest, followed by the proposed moment-based approach, and then the likelihood-based methods. In terms of statistical efficiency, exact maximum likelihood procedures and the proposed moment-based procedures perform best, followed by *lmer.split* and the *sgd* method. The *sgd* method is much faster than the other methods under consideration, but the computational gains are paid for with reduced statistical efficiency. The *mhglm* is not as fast as the *sgd* method, but in terms of statistical efficiency, it is competitive with the exact maximum likelihood methods.


\end{document}